\documentclass[11pt]{article}
\usepackage{amssymb}
\usepackage{amsmath}
\usepackage[algoruled,vlined]{algorithm2e}
\usepackage{graphicx}
\usepackage{times}
\usepackage{color}
\usepackage{bm}
\usepackage{framed}
\usepackage[left=1.0 in,right=1.0 in,top=1.0 in,bottom=1.0 in]{geometry}
\usepackage{xspace}
\newtheorem{theorem}{Theorem}
\newtheorem{definition}[theorem]{Definition}
\newtheorem{lemma}[theorem]{Lemma}
\newtheorem{corollary}[theorem]{Corollary}
\newtheorem{proposition}[theorem]{Proposition}
\newtheorem{example}{Example}
\newtheorem{fact}[theorem]{Fact}
\newtheorem{claim}[theorem]{Claim}
\newtheorem{remark}[theorem]{Remark}

\newcounter{step}
\newenvironment{step}{%
\addtocounter{step}{1}%
\begin{leftbar}\noindent\textbf {STEP \thestep :}\it%
}{\end{leftbar}}

\newenvironment{proof}{%
  \noindent{\it Proof\ }}{%
  \hspace*{\fill}$\square$
  \vspace{2ex}\\}
\newenvironment{proofof}[1]{%
  \vspace{2ex}
  \noindent{\it Proof of #1.\ }}{%
  \hspace*{\fill}$\square$
  \vspace{2ex}}

\newcommand{\true}{\mathsf{true}}
\newcommand{\false}{\mathsf{false}}

\newcommand{\calD}{\mathcal{D}}
 \newcommand{\calS}{\mathcal{S}}
 \newcommand{\calI}{\mathcal{I}}
 
 \newcommand{\calP}{\mathcal{P}}
 \newcommand{\Id}{\mathbf{I}} 
 
\newcommand{\abs}[1]{\vert #1 \vert}
\newcommand{\Abs}[1]{\left\vert #1 \right\vert}
\newcommand{\bra}[1]{\langle #1 |}

\newcommand{\ceil}[1]{\lceil #1 \rceil}
\newcommand{\Closedset}[1]{\left[ #1 \right]}
\newcommand{\Cset}{\Closedset}
\newcommand{\deq}{:=}
\newcommand{\ket}[1]{| #1 \rangle}
\newcommand{\Ket}[1]{\left| #1 \right\rangle}
\newcommand{\Openset}[1]{\left( #1 \right)}
\newcommand{\Oset}{\Openset}
\newcommand{\set}[1]{\{#1\}}
\newcommand{\Set}[1]{\left\{#1\right\}}
\newcommand{\Complex}{\mathbb{C}}
\newcommand{\Integer}{\mathbb{Z}}
\newcommand{\Natural}{\mathbb{N}}

\newcommand{\reg}[1]{\mathsf{#1}}
\newcommand{\e}{\mathrm{e}}
\newcommand{\im}{\mathbf{i}}

\newcommand{\LE}{\mathsf{LE}}
\newcommand{\SV}{\mathsf{SV}}
\newcommand{\Qhm}{$\mathrm{Q}_{h,m}$\xspace}

\newcommand{\consistent}{\mathsf{consistent}}
\newcommand{\inconsistent}{\mathsf{inconsistent}}
\newcommand{\vect}[1]{\mathbf{#1}}
\newcommand{\COLORCOUNT}{{\textsc{colorcount}}}
\newcommand{\CONSISTENCY}{{\textsc{consistency}}}
\newcommand{\QCONSISTENCY}{{\textsc{qconsistency}}}

\newcommand{\GHZSCALEDOWN}{{\textsc{ghz-scaledown}}}
\newcommand{\CNOT}{\textsc{cnot}}
\begin{document}

\title{A Fast Exact Quantum Algorithm for Solitude Verification}
\author{Seiichiro Tani\\
NTT Communication Science Laboratories \\
NTT Corporation, Japan.\\
{\tt tani.seiichiro@lab.ntt.co.jp}\\
}
\date{}
\maketitle

\begin{abstract}
Solitude verification
is arguably one of the simplest fundamental problems in distributed computing,
where the goal is to verify that there is a unique contender
in a network.
This paper devises a quantum algorithm that exactly solves the problem
on an anonymous network, which is known as a network model
with minimal assumptions [Angluin, STOC'80].
The algorithm runs
in $O(N)$ rounds if every party initially has the common knowledge
of an upper bound $N$ on the number of parties.
This implies that 
all solvable problems can be solved in $O(N)$ rounds on average
without error (i.e., with zero-sided error)
on the network.
As a generalization, a quantum algorithm that works in 
$O(N \log_2 (\max\set{k,2}))$ 
rounds
is obtained for the problem of exactly computing any symmetric Boolean function, 
over $n$ distributed input bits,
which is constant over all the $n$ bits whose sum is larger than $k$
for $k\in \set{0,1,\dots, N-1}$.
All these algorithms work with 
the bit complexities bounded by
a polynomial in $N$.
\end{abstract}

\pagestyle{plain}

\section{Introduction}
\subsection{Background}
In synchronous distributed models of computation, 
the number of rounds (also called the round complexity) 
is one of the most important complexity measures,
especially when we want to design fast distributed algorithms.
From a complexity-theoretic point of view,
seeking low-round complexity leads
to clarifying how much 
parallelism the problem inherently has.
This would be reminiscent of 
the study of shallow circuit classes (e.g.,  $\mathsf{NC}$),
in which
the depth of a circuit solving a problem
corresponds to the inherent parallelism of the problem.
In this paper, we study distributed algorithms with low-round complexity.

The round complexity is closely related to the diameter of the underlying graph of a given network.
This is because, when computing global properties of the network,
it necessarily 
takes at least as many rounds for message exchanges
as the value of the diameter for some party to get information from
the farthest party.
Therefore, 
when every party's initial knowledge includes
an upper bound $\Delta$ on the diameter,
the ultimate goal is
to achieve a round complexity
close to 
(typically, linear in) 
$\Delta$.
In particular, 
we are interested in whether the round complexity $O(n)$ 
(resp., $O(N)$)
can be achieved
when every party's initial knowledge
includes the number $n$ of parties
(resp.,  an upper bound $N$ on $n$).
This can actually be achieved in a straightforward manner
if there is a unique party (called the leader) distinguishable from the others
 or, almost equivalently, if every party has its own identity:
The unique leader, which can be elected
in $O(\Delta)$ rounds
if every party has its own identity,
can gather all the distributed inputs, solve the problem,
and distribute the solution to every party
in $O(\Delta)$ rounds.
However, it is not a simple task to bound the achievable round complexity
for networks where
no party has its own identity
(namely, all parties with the same number of communication links are identical).
Such a network is called 
an \emph{anonymous network}, 
which was first introduced by Angluin~\cite{Ang80STOC}
to examine how much each party in a network needs to know about its own identity and other parties'
(e.g., Refs.~\cite{ItaRod81FOCS,ItaRod90InfoComp,
AfeMat94InfoComp,KraKriBer94InfoComp,BolShaVigCodGemSim96ISTCS,YamKam96IEEETPDS-1,YamKam96IEEETPDS-2,BolVig02DM}),
and thereby understand the fundamental properties of distributed computing.
It has been revealed in the literature that anonymous networks make it 
highly non-trivial or even impossible to exactly solve many distributed problems, including the leader election one,
that are easy to solve on non-anonymous networks (i.e., the networks in which every party has its own identity).
Here, by ``exactly solve'', we mean ``solve without error within a bounded time''. 
The good news 
is that
if the number $n$ of parties is provided to each party,
all solvable problems can be solved exactly in $O(n)$ rounds%
\footnote{If the diameter $\delta$ is given,
all solvable problems can be solved in $O(\delta)$ rounds
~\cite{Hen14IEEETPDS}.}
for any unknown underlying graph
by constructing tree-shaped data structures,
called universal covers~\cite{Ang80STOC}
or views~\cite{YamKam96IEEETPDS-1}.
Obviously, however, this does not help us deal with the infinitely many
instances of fundamental problems
that are impossible to solve on 
anonymous networks.
The best known among them is
the leader election problem ($\LE_n$),
the problem of electing a unique leader.
There are infinitely many $n$ such that
$\LE_n$
cannot be solved exactly
for anonymous networks
with certain underlying graphs 
over $n$ nodes
even if $n$ is provided to each party~\cite{Ang80STOC,YamKam96IEEETPDS-1,BolVig02DM}.

The above situation changes drastically if quantum computation and communication are allowed
on anonymous networks (called \emph{anonymous quantum networks}):
$\LE_n$ can be solved exactly for 
any unknown underlying graph,
even when only an upper bound $N$ on $n$ is provided to 
every party~\cite{TanKobMat12TOCT}.
This implies that,
if a problem is solvable in non-anonymous networks,
then it is also solvable in anonymous quantum networks.
For the round complexity, however,
the known quantum algorithms for electing a unique leader 
require super-linear rounds in $N$ (or $n$ when $n$ is provided)
\cite{TanKobMat05STACS,TanKobMat12TOCT}.

Motivated by this situation, 
we study the linear-round exact solvability
of another fundamental problem,
the solitude verification problem ($\SV_n$)~\cite{AbrAdlHigKir86PODC,AbrAdlHigKir94JACM,HigKirAbrAdl97JALG},
in anonymous quantum networks.
The goal of $\SV_n$  is to verify that there is a unique contender
in a network with an unknown set of contenders (which may be empty) among the $n$ parties.
Although the final target is to clarify 
whether a unique leader can be \emph{elected} in linear rounds or not,
$\SV_n$ would be a natural choice as the first step.
This is because
$\SV_n$ is a subproblem of many common problems,
including the leader election problem:
a unique leader can be elected by repeating attrition and solitude verification
as observed in Ref.~\cite{AbrAdlHigKir86PODC}.
Another reason
is that
$\SV_n$ is one of the  simplest nontrivial problems concerned with the global properties of a given network,
as pointed out in Ref.~\cite{AbrAdlHigKir94JACM}.
Indeed, $\SV_n$ is not always solvable in the classical case:
One can easily show, 
by modifying the proof of Theorem 4.5 in Ref.~\cite{Ang80STOC},
that 
it is impossible to exactly solve $\SV_n$
on any anonymous classical network
whose underlying graph is not a tree if only an upper bound $N$ is provided
to each party
(the problem can be solved exactly in $O(n)$ rounds if $n$ is provided).
In the quantum setting,  the only quantum algorithms for $\SV_n$
are the straightforward ones that first elect a unique leader
(with a super-linear round complexity), 
who then verifies that there is a unique contender.

Recently,
Kobayashi et al.~\cite{KobMatTan14CJTCS}
proposed an $O(N)$-round quantum algorithm
when each communication link in the network is bidirectional, i.e.,
the underlying graph is undirected.
However, their algorithm \emph{cannot} work in the more general case
where the underlying graph is directed.
This is due to 
a technicality
that is distinctive in quantum computing:
Their algorithm uses a standard quantum technique to erase
``garbage'' information
generated in the course of computation.
The technique inverts some operations that have been performed
and thus involves
sending back messages
via bidirectional communication links
in the distributed computing setting.

\subsection{Our Results}
Let $\calD_n$ be the set of all  strongly connected digraphs with $n$ nodes.
Our main result is an $O(N)$-round quantum algorithm
that exactly solves $\SV_n$,
where the input to each party is 
a binary value indicating
whether the party is a contender or not
(see Sec.~\ref{sec:DefofSVandLE} for a more formal definition).

\begin{theorem}
\label{th:SV}
There exists a quantum algorithm that,
if an upper bound $N$ on the number $n$ of parties is provided
to each party,
exactly solves
$\SV_{n}$  in $O(N)$ rounds with bit complexity 
$\tilde{O}(N^8)$
on an anonymous network
with any unknown underlying graph in $\calD_n$.
\end{theorem}

As described previously, we are most interested in whether
$\LE_n$ can be solved exactly in $O(N)$ rounds.
For the present, we do not have an answer to this question.
We can, however, obtain a partial answer as  a corollary of 
Theorem~\ref{th:SV}:
There exists an $O(N)$-round \emph{zero-error} quantum algorithm for $\LE_n$.
Here, we say that a problem is solved with zero error
if there exists an algorithm that outputs a correct answer
with probability at least $1-\epsilon$
and gives up with probability at most $\epsilon$,
where $\epsilon$ is some non-negative constant less than $1$.
We can assume without loss of generality that $\epsilon$
is an arbitrarily small constant,
since a constant number of repetitions reduce the 
``give-up'' probability to an arbitrary small constant,
which changes the complexity by at most a constant factor.
\begin{corollary}
\label{cr:LE}
There exists a zero-error quantum algorithm that,
if an upper bound $N$ on the number $n$ of parties is 
provided to each party,
solves
$\LE_{n}$ 
in $O(N)$ rounds
with bit complexity 
$\tilde{O}(N^9)$
on an anonymous network
with any unknown underlying graph in $\calD_n$.
\end{corollary}
This implies that in the quantum setting, anonymous networks
can be converted to the corresponding non-anonymous ones
without error
in $O(N)$ rounds on average,
since a unique leader can assign a unique number to each party 
in $O(N)$ rounds (in the worst case).
In the special case of $N=O(n)$, 
if a classical problem is solvable 
in a non-anonymous classical/quantum network
in $O(n)$ rounds with a polynomial bit complexity,
then in the corresponding anonymous networks the problem is still solvable 
without error
in $O(n)$ rounds on average
with a polynomial bit complexity
whenever quantum computation
and communication are allowed.

We next consider a generalization of Theorem~\ref{th:SV}.
Note that we can think of $\SV_n$ as the problem of deciding whether
the Hamming weight (i.e., the sum) of the $n$ input bits is exactly one or not,
which is equivalent to computing the corresponding symmetric function.
As generalizations of this function, let us consider a
collection $\calS_{n}(k)$,
for $k\in \set{0,1,\dots, N-1}$,
of all  symmetric  Boolean functions $f\colon\set{0,1}^{n}\to \set{0,1}$
such that
$f(\vect{x})$ is constant over all $\vect{x}\deq (x_1,\dots, x_n)\in \set{0,1}^n$ with
$\sum_{i=1}^n x_i>k$.
Note that $\calS_{n}(k)\subset \calS_{n}(k+1)$ for each $k\in [0..n-2]$,
and $\calS_{n}(k)$ for any $k\ge n-1$ represents the set of all symmetric functions
over $n$ bits.
In particular, the function corresponding to $\SV_n$ 
belongs to $\calS_{n}(1)$.
We then have the following theorem.
\begin{theorem}
\label{th:symmetric}
Suppose that there are $n$ parties on an anonymous network 
with any unknown underlying graph in $\calD_n$
in which an upper bound $N$ on $n$ is provided to each party.
For every $f\in \calS_{n}(k)$ with $k\in \set{0,1,\dots, N-1}$,
there exists a quantum algorithm that exactly computes
$f(\vect{x})$
over distributed input $\vect{x}\in\set{0,1}^n$
on the network
in $O(N \log_2 (\max\set{k,2}))$  rounds
with a bit complexity bounded by some polynomial in $N$.
\end{theorem}
Note  that 
an $O(N)$-round quantum algorithm for $\LE_n$ 
would imply that all solvable problems, including computing $\calS_n (k)$,
can be solved in $O(N)$ rounds 
(since the leader can convert the anonymous network
into the corresponding non-anonymous one).
Computing $\calS_n(k)$ is thus something lying between $\SV_n$ and $\LE_n$
with respect to  linear-round solvability.

\subsection{Technical Outline}
Recall that
the reason the $O(N)$-round leader election algorithm in 
Ref.~\cite{KobMatTan14CJTCS}
does not work on directed graphs is that
it sends back 
messages
via bidirectional communication links
to erase ``garbage'' information produced in the course of computation.
This seems inevitable as it uses
(a version of) the quantum amplitude amplification~\cite{ChiKim98QCQC}
(or a special case of the general quantum amplitude amplification~\cite{BraHoyMosTap02AMS}).
It is a critical issue, however,  when the underlying graph is directed,
since,
although strong connectivity ensures at least one directed path
on which the message could be sent back, 
parties cannot identify the path
in the anonymous network (since the original sender of the message cannot be identified).
In the classical setting, such an issue cannot arise since
the message need not be sent back
(the sender has only to keep a copy of the message if it needs to).

Our idea for resolving this issue is to employ
the symmetry-breaking procedure introduced
in \cite[Sec.~4]{TanKobMat12TOCT}.
The procedure
was used to solve the leader election problem as follows:
Initially, all parties are candidates for the leader
and they 
repeatedly perform
a certain distributed procedure
that reduces the set $S$ of the candidates by at least 1.
More concretely, the procedure partitions $S$ into at least two subsets
if $\abs{S}\ge 2$ and removes one of them from $S$.
A simple but effective way of viewing this
is that
it
not only reduces $S$ but 
decides
whether $\abs{S}$ is at least two or not,
since it can partition $S$
only when there are at least two candidates.
This observation
would exactly solve $\SV_n$ by regarding $S$ as the set of contenders
if the procedure outputs the correct answer with certainty.
However, the procedure
heavily depends on 
the following unknowns:
the cardinality of $S$ and the number $n$ of parties.
In Ref.~\cite{KobMatTan14CJTCS},
a similar problem arises 
when deciding whether an $n$-bit string $\vect{x}$ is of Hamming weight at most $1$,
and it is resolved
by running a base algorithm in parallel
for all possible guesses at $\abs{\vect{x}}$
and making the decision based on 
the set of all outputs
(the ``base algorithm'' uses amplitude amplification
and is totally different from the symmetry-breaking approach).
Together with a simple algorithm for testing whether $\vect{x}$ is the all-zero string,
the parallel execution of the base algorithm is used in the leader election algorithm~\cite{KobMatTan14CJTCS} to verify that a random $\vect{x}$ is of  Hamming weight exactly one.
This verification framework actually works in our case, and
it underlies the entire structure of our algorithm.
Namely, we
replace the base algorithm in the framework
with
a subroutine
constructed 
by carefully combining the symmetry-breaking procedure introduced in Ref.~\cite{TanKobMat12TOCT}
with classical techniques related to
the view~\cite{YamKam96IEEETPDS-1,Nor95DAM,BolVig02DM,Tan12IEEETPDS}.
This means that all parties collaborate to perform this subroutine
in parallel 
for all possible pairs of guesses at $(n,\abs{S})$.
The round complexity is thus equal to 
the number of rounds required to perform the subroutine once, i.e.,
$O(N)$ rounds.
To show the correctness,
we prove that
the set of the outputs 
over all possible pairs of the guesses yields the correct answer to  any $\SV_n$ instance with certainty.
This needs an in-depth and careful analysis
of  
all operations of which our algorithm consists
for every pair not necessarily equal to $(n,\abs{S})$.

Before the present work, it has seemed as if
the symmetry-breaking approach introduced in 
Refs.~\cite{TanKobMat05STACS,TanKobMat12TOCT}
is entirely different from the amplitude amplification approach
used in Ref.~\cite{KobMatTan14CJTCS}.
Our algorithm first demonstrates that 
these approaches  are quite compatible, and, indeed,
the technical core of 
Refs.~\cite{TanKobMat05STACS,TanKobMat12TOCT}
can effectively function
in the algorithmic framework proposed in Ref.~\cite{KobMatTan14CJTCS}.
This would contribute to a better understanding of distributed quantum computing and would be very helpful for future studies of quantum algorithms.

Our algorithm can be
 generalized to the case of computing a family $\calS_n(k)$ of more general symmetric functions as follows:
All parties collaborate to partition $S$ 
into subsets
by recursively applying the procedure 
up to $\ceil{\log_2 \max\set{k,2}}$ levels. 
If there is a singleton set among the subsets at a certain recursion level,
then the algorithm stops and all parties elect the only member of the subset as a leader,
who can compute $\abs{S}$
and thus compute the value of the given function in $\calS_k$.
If no singleton set appears even after 
the$\ceil{\log_2 \max\set{k,2}}$-th recursion level,
there must be more than $k$ parties
in $S$, in which case any function in $\calS_k$ is constant by the definition.

\subsection{Related Work}
Pal, Singh, and Kumar~\cite{PalSinKum03ARXIV}
and D'Hondt and Panangaden~\cite{DHoPan06QIC} dealt with
$\LE_n$ and the GHZ-state sharing problem
in a different setting, where pre-shared entanglement is assumed
but only classical communication is allowed.
The relation between several network models that differ in available quantum resources
has been discussed 
by Gavoille, Kosowski, and Markiewicz~\cite{GavKosMar09DISC}.
Recently, 
Elkin et al.~\cite{ElkKlaNanPan14PODC}
proved that quantum communication cannot substantially speed up
algorithms for some fundamental problems, such as the minimum spanning tree,
compared to the classical setting.
For fault-tolerant distributed quantum computing,
the Byzantine agreement problem
and the consensus problem were studied
by Ben-Or and Hassidim~\cite{BenHas05STOC} 
and
Chlebus, Kowalski, and Strojnowski~\cite{ChlKowStr10DISC}, respectively.
In the cryptographic context where there are cheating parties,
Refs.~\cite{AhaSil10NJP,Gan09ARXIV} devises quantum algorithms
that elect a unique leader with a small bias.
Some quantum distributed protocols were experimentally demonstrated
by Gaertner et al.~\cite{GaeBouKurCabWei08PRL}
and Okubo et al.~\cite{OkuWanJiaTanTom08PRA}.

See  the surveys~\cite{BuhRoh03MFCS,DenPan06SIGACT,BroTap08SIGACTNews}
and the references therein for more  work on distributed quantum computing.

\subsection{Organization}
Section 2 defines the network model and the problems considered in this paper. It then mentions several known facts employed in the subsequent sections.
Section 3 provides the structure of our algorithm and then proves Theorem~\ref{th:SV} assuming several properties of the key subroutine \Qhm.
Section 4 describes \Qhm step by step
and 
presents numerous claims and propositions
to show 
how each step takes effect.
Section 5 proves all the claims and propositions appearing in Section~4
and then completes the proof that \Qhm has the properties assumed in Section~3.
Section 6 proves Corollary~\ref{cr:LE},
and then generalizes Theorem~\ref{th:SV} for proving Theorem~\ref{th:symmetric}.

\section{Preliminaries}
Let $\Complex$ be the set of all complex numbers,
$\Natural$ the set of all positive integers,
and $\Integer^+$ the set of all non-negative integers.
For any $m,n\in \Integer^+$ with $m<n$,
 $[m..n]$ denotes the set $\set{m,m+1,\dots ,n}$,
and $[n]$ represents $[1..n]$.

\subsection{\bf The Distributed Computing Model}
\label{subsec:model}
We first define a classical model
and then adapt it to the quantum model,
where every party can perform quantum computation and communication.

A classical \emph{distributed network} consists of multiple parties and
\emph{unidirectional} communication links, each of which connects a pair of parties.
By regarding the
parties and links as nodes and edges, respectively, in a graph,
the topology of the distributed network can be represented by a 
strongly connected digraph (i.e., directed graph),
which may have multiple edges or self-loops.

A natural assumption is that 
every party can distinguish one link from another
among all communication links incident to the party;
namely, it can assign a unique label  to  every such link.
We associate  these labels with \emph{communication ports}.
Since every party has \emph{incoming} and \emph{outgoing} communication
links
(although a self-loop is a single communication link,
it looks like a pair of incoming and outgoing links for the party), 
it has two kinds of \emph{communication ports} accordingly: \emph{in-ports} and \emph{out-ports}.
For a more formal definition in graph theory terms, we modify the 
definition given in  Ref.~\cite{Ang80STOC,YamKam96IEEETPDS-1},
which assumes undirected graphs:
the underlying digraph ${G\deq (V,E)}$ of the distributed network has a \emph{port-numbering},
 which is a set 
 of paired functions 
 $\{ (\sigma_{v}^{\text{in}},\sigma_{v}^{\text{out}}) \colon v\in V\}$
 such that for each node $v$ with in-degree $d_v^{\text{in}}$ and
 out-degree $d_v^{\text{out}}$,
 the function
$\sigma_{v}^{\text{in}}$ (resp., $\sigma_{v}^{\text{out}}$) 
 is a bijective map from the set of incoming edges (resp., outgoing edges) incident to $v$
 to
 the set
 $[d_{v}^{\text{in}}]$ (resp., $[d_{v}^{\text{out}}]$). 
 It is stressed that
 each function pair $(\sigma_{v}^{\text{in}},\sigma_{v}^{\text{out}})$ is defined 
 independently of any other pair $(\sigma_{v'}^{\text{in}},\sigma_{v'}^{\text{out}})$ with $v'\neq v$.
For the sake of convenience, we assume that
each edge $e\deq (u,v)\in E$ is labeled with 
the pair of the associated out-port 
 and in-port (of the two different parties);
namely,
 $(\sigma_u^{\text{out}}[e],\sigma_v^{\text{in}}[e])$
 (each party can know the labels of edges incident to it 
 by only a one-round message exchange as described in Example~\ref{example}).
In our model,
each party knows the number of its in-ports 
and out-ports
and can choose one of its in-ports
or one of its out-ports in any way
whenever it sends or receives a message.

In distributed computing, 
what information each party initially possesses
has a great impact on complexity.
Let $\calI_l$
be the information that only party $l$ initially knows,
such as its local state and the number of its ports.
Let  $\calI_G$ be
the information initially shared by all parties.
We may call $\calI_l$ and  $\calI_G$ \emph{local} and \emph{global} information, respectively.%
\footnote{
In this paper, we do not consider 
the case where 
only a subset of parties share
some information,
since
an upper bound on 
the complexity for that case
can be obtained by
regarding such information as local information
(when dealing with non-cryptographic/non-fault-tolerant problems, 
which is our case).}

Without loss of generality,
we assume that
every party $l$ runs the same algorithm
with $(\calI_l, \calI_G)$ as its arguments (or the input to the algorithm),
in addition to the instance $x_l$ of a problem to solve.
We will not explicitly write $(\calI_l, \calI_G)$ as input to algorithms
when it is clear from the context.
Note that 
$(\calI_l, \calI_G)$ is not part of the problem instance
but part of the model.
Also note that the algorithm may invoke subroutines
with part of $(x_l, \calI_l, \calI_G)$ as \emph{input to the subroutine}.
If all parties in a network have the same local information except
for the number of their ports,
the network is said to be \emph{anonymous}, and
the parties in the anonymous network are said to be anonymous.
In this paper, 
we assume that 
for each party $l$,
$\calI_l$
consists of a common initial state,
a common description of the same algorithm,
 and the numbers $(d_l^\text{in}, d_l^\text{out})$ of in-ports and out-ports.
An extreme case of networks is a regular graph, such as a directed ring,
in which case each party is identical to any other party;
that is, effectively,
every party has the same identifier.
Obviously, the difficulty of solving a problem depends on the underlying graph.
Moreover, it may also depend on port-numberings.
This
can be intuitively understood from Example~\ref{example}.
When solving problems on distributed networks,
we do not assume a particular port-numbering;
in other words, we say that a problem can be solved
if  there is an algorithm that solves the problem
for \emph{any} port-numbering.

This paper deals with only anonymous networks but
may refer to a party with its index (e.g., party $i$)
only for the purpose of clear descriptions.
Our goal is to construct an algorithm
that works 
for \emph{any} port-numbering on \emph{any} digraph in $\calD_n$,
where $\calD_n$ denotes the set of all $n$-node strongly connected digraphs, which may have multiple edges or self-loops, and is used through this paper.

\begin{example}
\label{example}
\emph{Fig.}~$\ref{fig:portnumbering}$ shows two anonymous networks,
$(a)$ and $(b)$,
on the same four-node regular graph with different port-numberings 
$\sigma $ and $\tau$, respectively,
where
$(1)$ each party has two in-ports and two out-ports, and
$(2)$ each directed edge $e\deq (u,v)$ is labeled with 
$(\sigma_{u}^{\emph{out}}(e),\sigma_{v}^{\emph{in}}(e))$
and 
$(\tau_{u}^{\emph{out}}(e),\tau_{v}^{\emph{in}}(e))$
on networks $(a)$ and $(b)$,
respectively,
where $\sigma_{u}^{\emph{out}}(e)$
$($resp., $\tau_{u}^{\emph{out}}(e)$$)$
 is put on the source side and 
 $\sigma_{v}^{\emph{in}}(e)$ 
$($resp., $\tau_{v}^{\emph{in}}(e)$$)$
 is put on the destination side.
Observe that each party can know the label of each incoming edge
incident to the party by exchanging a message$:$
Each party sends a message
``$i$" out via every out-port $i$, 
and if another party receives this message via in-port  $j$,
the receiver concludes that it has an incoming edge
with label $(i,j)$.
To elect a unique leader on network $(a)$ in \emph{Fig.}~$\ref{fig:portnumbering}$,
consider the following game: 
$(1)$ If a party has an incoming edge
with label $(i,j)$, then 
it scores 1 point if $i>j$ $($win$)$,
 0 points if $i=j$ $($draw$)$,
 -1 point if $i<j$ $($lose$)$$;$
$(2)$ each party earns the the sum of points over all its incoming 
edge$;$
and $(3)$ a party wins the game if it earns the largest sum of points 
among all parties.
It is easy to see that each party can compute
the sum of points as we observed.
In the case of $(a)$, 
the upper-left party is the unique winner:
it earns 1 point in total,
while the others earn $0$ or $-1$ points.
This fact can be used to elect a unique leader.
In the case of $(b)$, however, all parties earn $0$ points.
Hence, the above game cannot elect a unique leader.
Actually, no deterministic algorithm can elect a unique leader
in the case of $(b)$~\emph{\cite{YamKam96IEEETPDS-1}}.
\begin{figure}[htbp]
   \centering
   \includegraphics[scale=0.6]{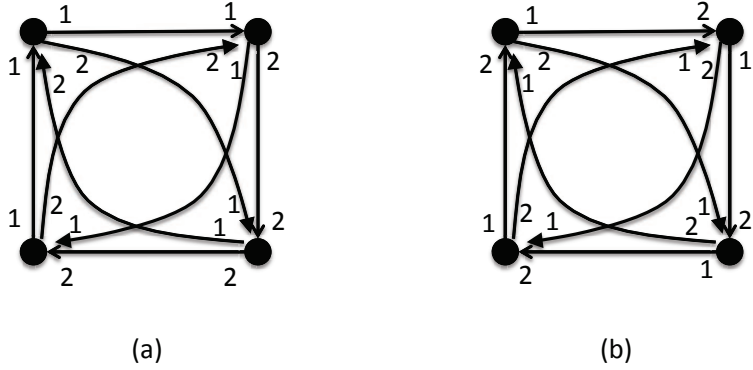}
   \caption{Anonymous networks on the same underlying graph with different port-numberings}
   \label{fig:portnumbering}
\end{figure}
\end{example}

A network is either
\emph{synchronous} or \emph{asynchronous}.
In the synchronous case, message passing is performed synchronously.
The unit interval of synchronization is called a \textit{round},
which consists of the following sequential execution of the two (probabilistic) procedures that are defined in the algorithm invoked by each party~\cite{Lyn96Book}:
one procedure changes the local state of the party 
depending on the current local state and the incoming messages, and then removes the messages from ports;
the other procedure then prepares new messages and decides the ports 
through which the messages 
should be sent, depending on  the current local state, and finally 
the messages are sent out via the ports. 
We do not impose any limit on the number of bits
in a message sent in each round.

A network that is not synchronous is \emph{asynchronous}.
In asynchronous networks,
every party can send messages at any time and 
the time it takes for
a message to go through a communication link 
is finite but not bounded.
This paper deals with synchronous networks
for simplicity, 
but 
our algorithms can be emulated in asynchronous networks without sacrificing the communication cost by just delaying the local operation that would be done in each round in the synchronous setting until a message arrives at every port.

The only difference between the quantum and classical models
is that every party can perform quantum computation
and communication in the former model
[for the basics of quantum computation and communication,
we refer readers to standard textbooks (e.g., Refs.~\cite{NieChu00Book,KitSheVya02Book,KayLafMos07Book})].
More concretely, the two procedures
for producing messages and changing local states
are replaced with physically realizable super-operators
(i.e., a trace-preserving completely positive super-operator)
that act on the registers storing the local quantum state
and quantum messages received to produce 
new quantum messages
and a new local quantum state
and to specify port numbers. 
Accordingly, we assume that every communication 
link can transfer quantum messages.
For sending quantum messages at the end of each round, 
each party sends out one of its quantum registers
through the specified out-port.
The party then receives quantum registers from its neighbors 
at the beginning of the next round
and uses them for local quantum computation.

 This paper focuses on the required number of rounds as 
 the primary
  complexity
 measure (called \emph{round complexity}). This is often used as an
 approximate value of time complexity, which includes the time taken by
 local operations as well as the time taken by message exchanges.
 Although our primary goal is to construct 
 algorithms with low round complexities,
our algorithms all have bit complexities bounded by certain polynomials in the given upper bound on the number of parties
(\emph{or} polynomial bit complexities for short),
where the \emph{bit complexity} of an algorithm is the number of bits or qubits
communicated by the algorithm (a.k.a., communication complexity).

Finally,
we assume that there are no faulty parties and no faulty communication links.

\subsection{Solitude Verification and Leader Election}
\label{sec:DefofSVandLE}
Let $n\in \Natural$.
For any bit string $\vect{x}\in \set{0,1}^n$, let $\abs{\vect{x}}$ be
the Hamming weight of $\vect{x}$, i.e., the number of $1$'s in $\vect{x}$.
For any $G\deq (V,E)\in \calD_n$,
without loss of generality, we assume that $V$ is identified with the set $[n]$.

For any $k\in \Integer^+$,
let $H_k\colon \set{0,1}^n\to \set{\true,\false}$ 
be the symmetric Boolean function 
that is $\true$ if and only if $\abs{\vect{x}}$ is equal to $k$ for input $\vect{x}\in \set{0,1}^n$
distributed over $n$ parties (i.e., each party is in possession of one of the $n$ bits).
Note that $k$ may be larger than $n$, in which case
$H_{k}(\vect{x})$ is $\false$ for all $\vect{x}\in \set{0,1}^n$.
The function $H_{k}$ is hence well-defined even if each party does not know the integer $n$
(since the value of $H_{k}$ depends not on $n$ but on $k$ and $\abs{\vect{x}}$).
We also define
$T_k\colon \set{0,1}^n\to \set{\true,\false}$ 
as the symmetric Boolean function 
such that 
$T_k(\vect{x})$ is $\true$ if and only if $\abs{\vect{x}}\le k$ 
for $\vect{x}\in \set{0,1}^n$;
namely, $T_k(\vect{x})=\vee _{i=0}^k H_i(\vect{x})$.
Note again that $T_{k}$ is well-defined even if $n$ is unknown.
The \emph{solitude verification problem} is equivalent
to computing $H_1$ as can be seen from the following definition.
\begin{definition}[\boldmath{Solitude Verification Problem ($\SV _n$)}]
Suppose that there is 
a distributed network 
with any underlying graph $G\in \calD_n$,
which is unknown to any party.
Suppose further that
each party
$i\in [n]$ in the network is given as input
a Boolean value  $x_i\in \set{0,1}$
and a variable $y_i\in \set{\true, \false}$ initialized to $\true$.
The goal is to set $y_i=H_1(\vect{x})$ 
for every $i\in [n]$,  
where $\vect{x}\deq (x_1,\dots, x_n)$.
\end{definition}
If every party has a unique identifier picked from, say, $[n]$,
this problem can easily be solved by simply gathering all $x_i$'s to
the party numbered $1$ (although the complexity may not be optimal).
On anonymous networks, however, this simple idea can no longer work
since the parties do not have a unique identifier.
If the global information $\calI_G$
includes
the \emph{exact} number $n$ of parties, $\SV_{n}$ can be solved deterministically 
by a non-trivial algorithm,
which runs
in
$O(n)$ rounds with a polynomial bit complexity~\cite{YamKam96IEEETPDS-1,BolVig02DM,Tan12IEEETPDS}. 
In a more general case, however, this is impossible.
\begin{fact}[\cite{Ang80STOC,ItaRod90InfoComp}]
\label{fact:SV}
There are infinitely 
many $n\in \Natural$
such that,
if 
only an upper bound on the number $n$ of the parties
is provided to each party as global information,
$\SV_{n}$ cannot be
solved
in the  zero-error $($i.e., Las Vegas$)$ setting as well as in the exact setting on anonymous classical networks with a certain underlying graph $G\in \calD_n$.
\end{fact}
Our main contribution is a quantum algorithm that
exactly computes the function 
$T_1$ in rounds linear in $N$ 
even if only an upper bound $N$ on $n$ is provided.
This implies that
there exists a quantum algorithm
that exactly computes $H_1=T_1\wedge \neg T_0$ (and thus $\SV_n$)  in rounds linear in $N$,
since there is a simple (classical) deterministic algorithm for computing $T_0$ 
(i.e., the negation of $\mathrm{OR}$ over all input bits).

We next define the leader election problem, which is
closely related to $\SV_n$.
\begin{definition}[\boldmath{Leader Election problem ($\LE _n$)}]
Suppose that there is 
a distributed network 
with any underlying graph $G\in \calD_n$,
which is unknown to any party.
Suppose further that
each party
$i\in[n]$
is given
a variable
$y_i$ initialized to $1$.
The goal is to set ${y_k=1}$ 
for arbitrary but unique $k\in [n]$
and ${y_i=0}$ for every remaining $i\in [n]\setminus \{k\}$.
\end{definition}
$\LE_n$ is a fundamental problem
with a long history of research starting 
from the dawn of distributed computing;
there are a lot of studies on efficient algorithms
for solving it on non-anonymous networks.
On anonymous networks, however,
it is impossible to
exactly solve $\LE_n$.
\begin{fact}[\cite{Ang80STOC,ItaRod90InfoComp,YamKam96IEEETPDS-1,BolVig02DM}]
\label{fact:LE}
There are  infinitely 
many $n\in \Natural$
such that,
even if the number $n$ is provided to each party
as global information,
$\LE_{n}$ 
cannot be solved exactly
on anonymous classical networks with a certain underlying graph $G\in \calD_n$.
Moreover, if 
only an upper bound on $n$ is provided to each party,
it is impossible to solve $\LE_n$ 
even with zero-error.
\end{fact}
Actually, the former part of this fact is a corollary of a more general theorem
proved in 
Ref.~\cite{YamKam96IEEETPDS-1,BolVig02DM},
which provides
a necessary and sufficient condition 
on underlying graphs and port-numbering
for exactly solving $\LE_n$
when every party knows the number $n$ as its global information.
The latter part of Fact~\ref{fact:LE}
(i.e., the zero-error unsolvability of $\LE_n$)
follows from Fact~\ref{fact:SV}
and the fact\footnote{
Once elected,
the leader can verify that $\abs{\vect{x}}$ is one: 
The leader first assigns a unique identifier to each party.
It then gathers all $x_i$ together with
the identifier of the owner of $x_i$ 
(along a spanning tree after setting it up).
}
that $\SV_n$ is reducible to $\LE_n$.

In contrast,  it is \emph{possible} to solve the problems
on anonymous \emph{quantum} networks.
\begin{fact}[\cite{TanKobMat05STACS}]
\label{fact:STACS05}
There exists a quantum algorithm
that,
for every $n\in \Natural$ with $n\ge 2$,
if an upper bound $N$ on $n$ is provided to each party
as global information, 
\emph{exactly} solves $\LE_{n}$  in $\Theta (N\log N)$ rounds
with a polynomially bounded bit complexity
$\tilde{O}(N \cdot p(N))$
on an anonymous network with any unknown underlying graph $G\in \calD_n$,
where $p(N)$ is the bit complexity of constructing 
the view~\emph{\cite{YamKam96IEEETPDS-1,BolVig02DM}}, a tree-like data-structure,
of depth $O(N)$.
The best known bound on
$p(N)$ is $\tilde{O}(N^6)$~\emph{\cite{Tan12IEEETPDS}}.
\end{fact}
\begin{remark}
When an upper bound $\Delta$ on
the diameter of the graph 
is provided to each party,
the round complexity becomes $O(\Delta (\log N)^2)$
by using the recent result on the view in Ref.~\emph{\cite{Hen14IEEETPDS}}.
\end{remark}

We next provide a powerful primitive in classical distributed computing:
a linear-round
algorithm
that deterministically computes any symmetric function
on anonymous networks.
This algorithm is actually obtained from a more generic one
that effectively uses the full-power of deterministic computation
in the anonymous classical network:
during the execution of the generic algorithm, every party constructs a tree-like data structure,
called \emph{view}~\cite{YamKam96IEEETPDS-1,BolVig02DM}, which contains as much information as it can gather
in the anonymous network.
In terms of graph theory,
the view of depth $k$ is defined for each node $v$,
and it is a labeled tree rooted at $v$ that is obtained 
by sharing the maximal common suffix of every pair of
$k$-length directed paths 
to the node $v$.
It is not difficult to see that
the view of depth $k$ can be constructed by 
exchanging messages $k$ times~\cite{YamKam96IEEETPDS-1,BolVig02DM} as follows:
(1) every party creates a $0$-depth view, which is nothing but a single node, $r$, labeled by the input to the party;
(2) for each $j=1,\dots, k$ in this order,
each party sends its $(j-1)$-depth view to every neighbor,
receives a $(j-1)$-depth view from every neighbor,
and then makes a $j$-depth view by connecting the node $r$ with
the roots of received views
by edges
labeled with port numbers.
Since it is proved in Ref.~\cite{Nor95DAM}
that setting $k=2n-1$ is necessary and sufficient to gather all the information in the network,
$(2n-1)$-depth views need to be constructed in general for solving problems.
For $k=2n-1$,
the above na\"{i}ve construction algorithm obviously has an exponential bit complexity in $n$, but
it is actually possible to compress each message so that the total bit complexity is
polynomially bounded in $n$~\cite{Tan12IEEETPDS}.
Since the Hamming weight $\abs{\vect{x}}$ of  $n$ distributed input bits
can be locally computed as a rational function of the number $n$ and the number of non-isomorphic subtrees
of depth $n-1$ in the view~\cite{YamKam96IEEETPDS-1,Nor95DAM,BolVig02DM},
every party can compute any symmetric function on $\vect{x}$
from its view whenever $n$ is given.
This algorithm is summarized as follows.
\begin{fact}[Computing Symmetric Functions~\cite{YamKam96IEEETPDS-1,Nor95DAM,BolVig02DM,Tan12IEEETPDS}]
\label{th:view}
For any $n\in \Natural$ with $n\ge 2$,
suppose 
an anonymous network with any unknown underlying graph $G\in \calD_n$,
where the number $n$ of parties is given to each party as global information.
Then,
there exists a deterministic algorithm that computes
any symmetric function over $n$ distributed input bits
in $O(n)$ rounds 
with bit complexity $O(p(n))$,
where $p(\cdot)$ is the function defined in Fact~$\ref{fact:STACS05}$.
\end{fact}
To exactly compute the symmetric function with this algorithm,
every party needs to know
the exact number $n$ of the parties.
Nevertheless,
even when a wrong $m\in \Natural$  is provided instead of $n$,
the algorithm can still run through and output some value.
Namely, it constructs the view of depth $2m-1$ for every party 
and outputs the value 
of the rational function over the number $m$ and the number of isomorphic subtrees in the view, as can be seen from the above sketch of the algorithm.
This requires
$O(m)$ rounds and 
$O(\max\set{p(m),p(n)})$ bits of communication.
In fact, we will use this algorithm for a 
guess $m$ at $n$ (this $m$ is not necessarily equal to $n$).
Although the output may be wrong, 
the set of the outputs over all possible guesses $m$'s
contains useful information
as will be described in the following sections.

Finally, we define some terms.
Suppose that each party $l$ has 
a
$c$-bit string {$z_l \in \set{0,1}^{c}$}
(i.e., the $n$ parties share a $cn$-bit string ${\vect{z} \deq (z_1,z_2, \cdots ,z_n)}$).
Given a set ${S \subseteq [n]}$,
the string $\vect{z}$ is said to be \emph{consistent} over $S$
if $z_l$ 
has
the same value for all $l$ in $S$.
Otherwise, $\vect{z}$ is said to be \emph{inconsistent} over $S$.
In particular, if $S$ is the empty set, then any string $\vect{z}$ is
consistent over $S$.
We also say that a $cn$-qubit pure state
${\ket{\psi} \deq \sum_{\vect{z}\in \set{0,1}^{cn}} \alpha_\vect{z} \ket{\vect{z}}}$
shared by the $n$ parties is \emph{consistent (inconsistent)} over $S$
if ${\alpha_\vect{z} \neq 0}$ only for $\vect{z}$'s that are consistent
(inconsistent) over $S$. 
Note that there are pure states that are neither consistent nor
inconsistent over $S$
(i.e., {superpositions of}
both consistent string(s) and inconsistent string(s) over $S$).
We may simply say ``consistent/inconsistent strings/states''
if the associated set $S$ is clear from 
the context.
We say that a quantum state $\ket{\psi}$
is an \emph{$m$-partite GHZ-state} if $\ket{\psi}$ is
of the form 
$\frac{1}{\sqrt{2}}(\ket{0}^{\otimes m}+\ket{1}^{\otimes m})$ for some
natural number $m$.
When $m$ is clear from the context, we may simply call the state $\ket{\psi}$ 
a \emph{GHZ-state}.

\section{Solitude Verification (Proof of Theorem~\ref{th:SV})}
\label{sec:SVAlgorithm}
This section proves Theorem~\ref{th:SV}
by showing
a quantum algorithm
that exactly computes $H_{1}$
on $n$ bits distributed 
over an anonymous quantum network with $n$ parties
when an upper bound $N$ on $n$ 
is provided as  global information.
Suppose that a bit $x_i$ is provided to 
each party $i\in [n]$ as input.
We say that any party $i$ with $x_{i}=1$ is \emph{active} and any party $j$ with 
$x_{j}=0$ is \emph{inactive}.
Let $\vect{x}\deq (x_{1},\dots, x_{n})$.

The algorithm  actually computes the functions
$T_{1}$ and  $T_{0}$ on input $\vect{x}$ in parallel
and then outputs $H_{1}(\vect{x})=\neg T_{0}(\vect{x})\wedge T_{1}(\vect{x})$.
The function $T_0$ can be computed in any anonymous network
by a simple and standard deterministic algorithm
in $O(N)$ rounds if any upper bound $N$ on $n$ is given,
as stated in Proposition~\ref{pr:T_{0}} (the proof is provided in Sec.~\ref{sec:proofs} for completeness).
\begin{proposition}[Computing {\boldmath{$T_0$}}]
\label{pr:T_{0}}
Suppose that  there are $n$ parties
on an anonymous 
 network with any underlying graph $G\deq (V,E)\in \calD_n$,
 in which 
 an upper bound $\Delta$ on $n$
is given as  global information.
 Then, there exists a deterministic algorithm
 that computes $\neg\bigvee_{i\in [n]}x_i$
 on the network
 in $O(\Delta)$ rounds with  bit complexity
$O(m\Delta)$, where $m\deq|E|$,
if every party $i\in V$ gives a bit $x_{i}\in \set{0,1}$ as the input to the algorithm.
\end{proposition}
The difficult part is computing $T_1$.
It runs another quantum algorithm
\Qhm
as a subroutine for all pairs of 
$(h,m)\in [0..m]\times [2..N]$.
The integers $h$ and $m$ mean
guesses at the number of active parties
(i.e., $\abs{\vect{x}}$) and the number of parties (i.e., $n$), respectively, which are unknown to any party.
The following lemma 
states the properties of \Qhm.
\begin{lemma}[Main]
\label{lm:Q{h,m}}
There exists a set of distributed quantum algorithms $\set{\mbox{\Qhm} \colon (h,m)\in [0..m]\times [2..N]}$ such that,
if every party $i$ $(i=1,\dots, n)$
performs
the algorithm \Qhm
with
a bit $x_i\in \set{0,1}$ and
an upper bound $N$ on the number $n$ of parties,
then \Qhm always reaches a halting state%
\footnote{
This property is crucial
when one emulates our algorithm in asynchronous
(anonymous) networks, 
since deadlocks that parties cannot detect may occur.}
in $O(N)$ rounds
with $O(N^3)$ qubits and $\tilde{O}(N^6)$ classical bits of communication
for
every $m\in [2..N]$ and
every $h\in [0..m]$
and satisfies the following properties$:$
\begin{enumerate}
\item For each $(h,m)$, \Qhm outputs ``$\true$'' or ``$\false$''
at each party, where these outputs agree over all parties$;$
\item
for $(h,m)=(\abs{\vect{x}},n)$,
\Qhm
 computes $T_1(\vect{x})$ with certainty
for every $\vect{x}$$;$ and
\item 
for every $m\in [2..N]$ and
every $h\in [0..m]$,
\Qhm 
outputs ``$\true$''
with certainty
whenever $\abs{\vect{x}}\in \set{0,1}$,
\end{enumerate}
where $\vect{x}\deq(x_{1},\dots, x_{n})$.
\end{lemma}
Section~\ref{sec:keysub}
demonstrates how \Qhm works,
and Sec.~\ref{sec:proofs}
proves that it works as stated in the lemma.

After
all parties collaborate to run \Qhm (in parallel)
for all pairs $(h,m)$,
every party obtains 
$(N-1)(m+1)$ 
classical results:
\[\set{(h,m,q_{h,m}) \colon (h,m)\in [0..m]\times [2..N],q_{h,m}\in \set{\true,\false} },\]
where $q_{h,m}$ is the output of  \Qhm for input $\vect{x}$.
Note that this set of results agree over all parties by property~1
in Lemma~\ref{lm:Q{h,m}}.
From properties 2 and 3,
 we can observe that
there is at least one $\false$ 
in  $\set{q_{h,m}}$ 
if and only if $T_1(\vect{x})$ is ``$\false$''
(since 
$q_{h,m}$ is equal to $T_1(\vect{x})$
whenever $(h, m)$ equals $(\abs{\vect{x}},n)$).
Thus, every party can locally compute $T_1(\vect{x})$,
once it obtains all the classical results $\set{q_{h,m}}$.
Therefore, together with the algorithm in 
Proposition~\ref{pr:T_{0}},
$H_{1}(\vect{x})=\neg T_{0}(\vect{x})\wedge T_{1}(\vect{x})$
can be computed exactly.
Moreover, 
Proposition~\ref{pr:T_{0}} and Lemma~\ref{lm:Q{h,m}}
imply that
the entire solitude verification algorithm runs 
in $O(N)$ rounds and
with $O(N^5)$ qubits 
and $\tilde{O}(N^8)$ classical bits of communication.
The pseudo code, 
Algorithm QSV,
provided below 
summarizes the above operations.
\begin{algorithm}[ht]
\small
\SetAlgoLined
\SetKw{Notation}{Notation}
\SetKw{Global}{Global Information:}

\KwIn{a classical variable $x_{i}\in \set{0,1}$ and $N\in \Natural$. }

\KwOut{$y\in \set{\true,\false}$.}
\BlankLine
\Begin{

\ForEach{$(h, m)\in  [0..m]\times [2..N]$}
{perform in parallel \Qhm with $(x_{i},N)$ to obtain an output $q_{h,m}$\;}

\lIf{$q_{h,m}=\false$ for at least one of the pairs $(h,m)$}{set $y_{1}:=\false$}

\lElse{set $y_{1}:=\true$}\;

compute $y_0:=T_0(\vect{x})$ with input $x_i$ and $N$\;

\Return{$y:=\lnot y_{0}\land y_{1}$}.
}
\SetAlgoRefName{\bf QSV}
\caption{Every party $i$ performs the following operations.} \label{QSV}
\end{algorithm}

\section{Algorithm $\mbox{\boldmath Q$_{h,m}$}$}
\label{sec:keysub}
To demonstrate simply
how 
\Qhm
works,
we mainly consider the case where $(h,m)$ equals $(\abs{\vect{x}},n)$
(and defer the analysis in the other cases until the following sections).
In fact,
\Qhm originates from 
the following simple idea.
Every active party flips a coin,  broadcasts the outcome, and
receives the outcomes of all parties.
Obviously, 
those outcomes cannot include both heads and tails
if there is at most \emph{one} active party.
Otherwise, they include both with high probability.
This implies that, with high probability,
one can tell whether 
the number of active parties
is at most one or not,
based on
whether the outcomes include both heads and tails.
One can thus obtain the correct answer to $\SV_n$ with high probability,
since one can easily check whether there exists at least one active party.
However, this still yields the wrong answer with positive probability.
Since we want to solve the problem exactly, we need to suppress this error.
For this purpose,
we carefully combine several quantum tools 
with classical techniques
based on
the idea sketched in \cite[Sec.~4]{TanKobMat12TOCT}.

Suppose that every party $i$ has a two-qubit register $\reg{R}_i$.
Let $\set{\ket{\emptyset},\ket{\set{0}},\ket{\set{1}},\ket{\set{0,1}}}$
be an orthonormal basis
of $\Complex^{4}$
(``$\emptyset$'' represents the empty set).
For notational simplicity, we also use $\hat{0}, \hat{1}, \times$ to
denote $\set{0}, \set{1}, \set{0,1}$, respectively
(``$\times$'' intuitively means \emph{inconsistency} since we have both $0$ and $1$ in the corresponding set).
Without loss of generality, we identify
$\ket{\hat{0}}, \ket{\hat{1}}, \ket{\emptyset},\ket{\times}$
with $\ket{00},\ket{01},\ket{10},\ket{11}$, respectively.
We now start to describe \Qhm  step by step
and explain the effect of each step by showing claims and propositions.
We defer their proofs to the next section.

\begin{step}
Every active party $i$
prepares 
$(\ket{\hat{0}}+\ket{\hat{1}})/ {\sqrt{2}}$
in  register $\reg{R}_i$,
while every inactive party $j$ prepares $\ket{\emptyset}$ 
in $\reg{R}_j$.
\end{step}

Let $\reg{R}_S$ be the set of $\reg{R}_i$ over all $i\in S$,
where $S$ is the set of active parties,
and let $\reg{R}_{\overline{S}}$ be the set of  $\reg{R}_i$ over all $i\in [n]\setminus S$.
We then denote by $\reg{R}$ the pair $(\reg{R}_{\overline{S}},\reg{R}_{S})$.
The quantum state  $\ket{\psi^{(0)}_{\vect{x}}}$ over $\reg{R}$
can be written as 
\[
\ket{\psi^{(0)}_{\vect{x}}}_{\reg{R}}=
 \Cset{
 \ket{\emptyset}^{\otimes n-\abs{\vect{x}}}
 }
 _{\reg{R}_{\overline{S}}}
\otimes
 \Cset{
\Oset{
   \frac{\ket{\hat{0}}+\ket{\hat{1}}}{\sqrt{2}}
 }^{\otimes \abs{\vect{x}}}
 }
 _{\reg{R}_S}.
\]
\begin{step}
All parties attempt to project $\ket{\psi^{(0)}_{\vect{x}}}$
onto the space spanned by 
$\set{\ket{\vect{z}}\colon \vect{z}\in A}$
or 
$\set{\ket{\vect{z}}\colon \vect{z}\in A^c}$,
where
$A$ is the set of all strings $\vect{z}$ in 
$\set{\hat{0},\hat{1},\emptyset}^{n}$
such that
$\vect{z}$ does not contain both $\hat{0}$ and $\hat{1}$ simultaneously,
i.e., $A\deq\set{\hat{0},\emptyset}^{n}\cup \set{\hat{1},\emptyset}^{n}$,
and $A^c\deq \set{\emptyset,\hat{0},\hat{1},\times}^{n}\setminus A$
is the complement of $A$.
\end{step}

For STEP~2, 
every party applies the operator $\Phi_\Delta$ with $\Delta \deq N$
to $\ket{\psi^{(0)}_{\vect{x}}}$,
where $\Phi_\Delta$ is defined in the following claim.
The resulting state 
$\ket{\psi^{(1)}_{\vect{x}}}_{(\reg{R},\reg{Y},\reg{G})}\deq \Phi_{N}\ket{\psi^{(0)}_{\vect{x}}}_{\reg{R}}$
is $\frac{1}{\sqrt{2^{\abs{\vect{x}}}}}$ times

\begin{equation}\label{eq:state1}
\sum_{\vect{z}\in A}\ket{\vect{z}}_{\reg{R}}\otimes(\ket{\consistent}^{\otimes n})_{\reg{Y}}\otimes \ket{g_N(\vect{z})}_{\reg{G}}+
\sum_{\vect{z}\in A^c\cap \set{\hat{0},\hat{1},\emptyset}^{n}}
\ket{\vect{z}}_{\reg{R}}\otimes(\ket{\inconsistent}^{\otimes n})_{\reg{Y}}\otimes \ket{g_N(\vect{z})}_{\reg{G}}.
\end{equation}

\begin{claim}
\label{claim:step2}
There exists a distributed quantum algorithm 
$\QCONSISTENCY$
that,
for given upper bound $\Delta$ on the diameter of the underlying graph $G\deq (V,E)\in \calD _n$,
implements the following operator%
\footnote{
More formally, this operator should be written as a unitary one
acting on ancillary registers initialized to $\ket{0}$ as well as
$\reg{R}_i.$}$:$
\begin{equation}
\Phi_{\Delta}\deq \sum_{\vect{z}\in A}\ket{\vect{z}}\bra{\vect{z}}_{\reg{R}}\otimes(\ket{\consistent}^{\otimes n})_{\reg{Y}}\otimes \ket{g_{\Delta}(\vect{z})}_{\reg{G}}+
\sum_{\vect{z}\in A^c}\ket{\vect{z}}\bra{\vect{z}}_{\reg{R}}\otimes(\ket{\inconsistent}^{\otimes n})_{\reg{Y}}\otimes \ket{g_{\Delta}(\vect{z})}_{\reg{G}},
\end{equation}
where
$\set{\ket{\consistent},\ket{\inconsistent}}$ is an orthonormal basis of $\Complex^2$,
$\reg{Y}\deq (\reg{Y_1},\ldots,\reg{Y_n})$ denotes a set of ancillary 
single-qubit registers $\reg{Y}_i$ possessed by party $i$ for all $i\in [n]$,
and $\reg{G}$ denotes a set of all the other ancillary registers,
the content $g_{\Delta}(\vect{z})$ of which is a bit string
 uniquely determined%
\footnote{
Actually, $g_{\Delta}(\vect{z})$ also depends on the underlying graph, 
but we assume without loss of generality
that
the graph is fixed.}
by $\Delta$ and $\vect{z}$.
Moreover, $\QCONSISTENCY$ runs in $O(\Delta)$ rounds
and communicates $O(\Delta \abs{E})$ qubits.
In particular, any upper bound $N$ of $n$ can be used as $\Delta$.
\end{claim}
Note that the operator $\Phi_{\Delta}$ has an effect similar to the measurement
$\set{\Pi_{0},\Pi_{1}}$
on $\ket{\psi^{(0)}_{\vect{x}}}$,
where $\Pi_{0}\deq \sum_{\vect{z}\in A}\ket{\vect{z}}\bra{\vect{z}}$
and $\Pi_{1}\deq \Id_{4^n}-\Pi_{0}$ for the identity operator $\Id_{4^n}$ over
the space $\Complex^{4^n}$.
The difference  is that 
$\Phi_{\Delta}$ leaves the garbage part 
$\ket{g_{\Delta}(\vect{z})}$, 
which is due to the ``distributed execution'' of the measurement.

\begin{step}
Every party $i$ measures the register $\reg{Y}_i$ of $ \ket{\psi^{(1)}_{\vect{x}}}_{(\reg{R},\reg{Y},\reg{G})}$ in the basis $\set{\ket{\consistent},\ket{\inconsistent}}$.
If the outcome is ``$\inconsistent$'',  then \Qhm halts and returns ``$\false$''.
\end{step}
To understand the consequence of each possible outcome,
we first provide an easy proposition.
\begin{claim}
\label{claim:step3-1}
$\ket{\psi^{(0)}_{\vect{x}}}_{\reg{R}}$ is in 
the space spanned by
$\set{\ket{\vect{z}}\colon \vect{z}\in A}$
if and only if $\abs{\vect{x}}\le 1$.
\end{claim}

If the outcome is ``$\inconsistent$'',
there must exist
$\vect{z}\in A^c$
such that
$\ket{\vect{z}}$ 
has a non-zero amplitude in $\ket{\psi^{(0)}_{\vect{x}}}_{\reg{R}}$,
implying that 
$\ket{\psi^{(0)}_{\vect{x}}}_{\reg{R}}$ 
does \emph{not} lie
in the space spanned by
$\set{\ket{\vect{z}}\colon \vect{z}\in A}$
and thus $\abs{\vect{x}}\ge2$ by Claim~\ref{claim:step3-1}.
Therefore, every party can conclude that
$T_1(\vect{x})$ is ``$\false$''
without error if
the outcome is ``$\inconsistent$''.

If the outcome is ``$\consistent$'', we have the resulting state
(from Eq.~(\ref{eq:state1})):
$$
\ket{\psi^{(2)}_{\vect{x}}}_{(\reg{R},\reg{G})}=
c\sum_{\vect{z}\in A}\alpha_{\vect{z}}^{(\vect{x})}\ket{\vect{z}}_{\reg{R}}\otimes \ket{g_{\Delta}(\vect{z})}_{\reg{G}},$$
where $c\deq 1/\sqrt{\sum_{\vect{z}\in A}
\bigl|{\alpha_{\vect{z}}^{(\vect{x})}}\bigr|
^2}$
is the normalizing factor.
In this case, the parties 
cannot determine the value of $T_1$ on $\vect{x}$
since $\ket{\psi^{(0)}_{\vect{x}}}_{\reg{R}}$ 
may or may not be in the space
$\set{\ket{\vect{z}}\colon \vect{z}\in A}$.
Hence, we need a few more steps.

For proceeding to the next step, the following observation 
is crucial, which is obtained from the definition of $A$:
In the state $\ket{\psi^{(2)}_{\vect{x}}}_{(\reg{R},\reg{G})}$,
the contents of the registers $\reg{R}_i$ for
all active parties $i$ 
(i.e., the values at the positions of all such $i$ in $\vect{z}\in A$)
are identical; namely, they are all ``$\hat{0}$'' or all ``$\hat{1}$''.
Furthermore, we can prove the following claim:
\begin{claim}
\label{claim:step3-2}
The state $\ket{\psi^{(2)}_{\vect{x}}}_{(\reg{R},\reg{G})}$
is
of the form
\begin{equation}
\label{eq:state2}
\frac{1}{\sqrt{2}}
\Oset{\ket{\hat{0}}^{\otimes p}+\ket{\hat{1}}^{\otimes p}}
_{(\reg{R}_S,\reg{G}')}
\otimes 
\Oset{\ket{\emptyset}^{\otimes q}}
_{(\reg{R}_{\bar{S}},\reg{G}'')}
\end{equation}
for  some $p,q\in \Integer^+$ 
with $p+q\ge n$,
where $(\reg{G}', \reg{G}'')$ is a certain partition of the set of registers in $\reg{G}$. 
In particular, if there are no active parties, then $p$ is equal to $0$
and thus $\ket{\psi^{(2)}_{\vect{x}}}_{(\reg{R},\reg{G})}$
is a tensor product of $\ket{\emptyset}$.
\end{claim}
This follows from the fact that
the operator $\Phi_{\Delta}$ 
consists of idempotent operations, copy operations via $\CNOT$,
and register exchanges.

\begin{step}
Every party attempts to shrink the \emph{GHZ}-state part of
$\ket{\psi^{(2)}_{\vect{x}}}_{(\reg{R},\reg{G})}$
down to a $\abs{\vect{x}}$-partite
\emph{GHZ}-state over only the registers $\reg{R}_i$ of active parties
by performing a subroutine, called $\GHZSCALEDOWN$, the properties of which are described in Claim~$\ref{claim:step4}$.
If $\GHZSCALEDOWN$ outputs ``$\true$'', then \Qhm halts and returns ``$\true$''.
\end{step}
For realizing STEP~4,
if the state $\ket{\psi^{(2)}_{\vect{x}}}_{\reg{R}}$
were over local registers, 
we could just apply the following standard technique:
measure all registers except all $\reg{R}_i$ of active parties in the Hadamard basis $\set{\ket{+},\ket{-}}$,
count the number of outcomes $\ket{-}$,
and correct the phase by rotating the state by the angle 0 or $\pi$, depending on (the parity of) the number.
By the definition of anonymous networks, however, 
every (active) party has to
use the \emph{same} algorithm
to collaboratively
compute the parity of the number of $\ket{-}$ and rotate the state if needed.
For this purpose, 
the classical algorithm 
provided in Fact~\ref{th:view}
seems suitable, since it computes
any symmetric function
in an \emph{anonymous} network.
What one should note here is that this classical algorithm outputs a correct answer
\emph{if} it is given the correct number $n$ of parties (as $m$).
If $m$ is equal to $n$,
the algorithm outputs
a correct answer. 
In general cases, however, 
the value of $m$ is  merely a guess at $n$ and not necessarily equal to $n$.
To rotate the state by the angle $\pi$, every active party
applies to its share of the registers the rotation operator
for the angle $\pi/\abs{\vect{x}}$,
so that the 
sum of the angles over all active parties is $\pi$.
This works correctly \emph{if} the number of active parties
is given (as $h$).
With the assumption that $h$ equals $\abs{\vect{x}}$, the state
should be rotated by the correct angle $\pi$.
In general cases, however, 
the value of $h$ is  not necessarily equal to $\abs{\vect{x}}$.

In the case where $\abs{\vect{x}}$ equals $0$,
no active parties exist and thus no operations are performed in this step.

The following claim summarizes the effect of $\GHZSCALEDOWN$.
\begin{claim}
\label{claim:step4}
There exists a distributed quantum algorithm $\GHZSCALEDOWN$ such that,
for given $(h,m)$, 
\begin{itemize}
\item if $\abs{\vect{x}}$ equals zero,  then for any $(h,m)$ it 
halts with the output ``$\true$'' at every party or applies the identity operator
to $\ket{\psi^{(2)}_{\vect{x}}}_{(\reg{R},\reg{G})}$$;$
\item
else if $(h,m)$ equals $(\abs{\vect{x}},n)$,
it transforms
$\ket{\psi^{(2)}_{\vect{x}}}_{(\reg{R},\reg{G})}$
to
$\ket{\psi^{(3)}_{\vect{x}}}_{\reg{R}_S}
\deq 
\frac{1}{\sqrt{2}}(\ket{\hat{0}}^{\otimes \abs{\vect{x}}}+\ket{\hat{1}}^{\otimes \abs{\vect{x}}})_{\reg{R}_S}$$;$
\item else 
it halts with the output ``$\true$'' at every party or
transforms
$\ket{\psi^{(2)}_{\vect{x}}}_{(\reg{R},\reg{G})}$
to the state 
$\ket{\tilde{\psi}^{(3)}_{\vect{x}}}_{\reg{R}_S}
\deq \frac{1}{\sqrt{2}}(\ket{\hat{0}}^{\otimes \abs{\vect{x}}}+e^{i\theta_{h,m}}\ket{\hat{1}}^{\otimes \abs{\vect{x}}})_{\reg{R}_S}$ 
for some real $\theta_{h,m}$.
\end{itemize}
Moreover, $\GHZSCALEDOWN$ runs in $O(m)$ rounds and communicates
$O(\max\set{p(m),p(n)})$ classical bits,
where $p(\cdot)$ is the function defined in Fact~$\ref{fact:STACS05}$.
\end{claim}
Then, \Qhm proceeds to STEP~5.
By Claim~\ref{claim:step4},
STEP~5 is performed with $\ket{\tilde{\psi}^{(3)}_{\vect{x}}}_{\reg{R}_S}$
only if both $(h,m)\neq (\abs{\vect{x}},n)$ and $\vect{x}\ge 1$ hold.
In this case, to meet item~3 in Lemma~\ref{lm:Q{h,m}},
we only need to examine the output of Q$_{h,m}$
on $\vect{x}$ with $\abs{\vect{x}}=1$
(we defer until the next section easy proofs that item~1 holds
and that Q$_{h,m}$ reaches a halting state for all $\vect{x}$).
For $\vect{x}$ with $\abs{\vect{x}}=1$,
the behavior of  Q$_{h,m}$ is almost the same as in the case of 
$(h,m)=(\abs{\vect{x}},n)$.
Thus, in the following part of this section, we assume that 
$\abs{\vect{x}}$ is equal to $0$ or STEP~5 starts with 
$\ket{{\psi}^{(3)}_{\vect{x}}}_{\reg{R}_S}$ for simplicity
(the proof of Lemma~\ref{lm:Q{h,m}} provided in the next section will 
rigorously analyze all cases).

\begin{step}
Every active party $i$ applies the local unitary operator $W_{h}$
to its register $\reg{R}_i$, 
where $W_h$ is 
defined 
as follows$:$
\begin{equation*}
W_h\deq
\begin{cases}
I_2 \otimes U_h  & \text{$h$ is even and at least two}\\
V_h\cdot \CNOT_{2\rightarrow 1}  & \text{$h$ is odd and at least three}\\
I_2\otimes I_2& \text{$h$ is zero or one},
\end{cases}
\end{equation*}
where
$U_h$ and $V_h$ are defined in Sec.~$3.3$ of 
Ref.~\emph{\cite{TanKobMat12TOCT}}, 
$\CNOT_{2\rightarrow 1}$ acts on the first qubit,
controlled by the second qubit,
and $I_2$ is the identity over $\Complex_2$.
\end{step}

To see how the operator $W_h$ works, 
we provide the following claim.
\begin{claim}[Adaptation from Lemmas 3.3 and 3.5 in Ref.~\cite{TanKobMat12TOCT}]
\label{claim:step5} \emph{STEP~5} satisfies the following:
\begin{itemize}
\item 
If $\abs{\vect{x}}$ is zero, then 
\emph{STEP 5} is effectively skipped.

\item 
If $\abs{\vect{x}}$ is one, then 
$\ket{\psi^{(4)}_{\vect{x}}}_{\reg{R}_S}
\deq W_h\ket{\psi^{(3)}_{\vect{x}}}_{\reg{R}_S}
=
(\alpha \ket{\hat{0}}
+\beta\ket{\hat{1}}
+\gamma \ket{\emptyset}
+\delta\ket{\times})
_{\reg{R}_S}
$ 
for some $\alpha, \beta,\gamma, \delta \in \Complex$.

\item 
If $\abs{\vect{x}}$ is at least two and
$h$ equals $\abs{\vect{x}}$,
then it holds that
\[
\ket{\psi^{(4)}_{\vect{x}}}_{\reg{R}_S}
\deq W_h^{\otimes \abs{\vect{x}}}\ket{\psi^{(3)}_{\vect{x}}}_{\reg{R}_S}
=
\sum_{\vect{y}\in 
B
}
\beta_{\vect{y}}\ket{\vect{y}}_{\reg{R}_S},
\]
where
$ B\deq \Set{\ket{\vect{z}}\colon \vect{z}\in \set{\hat{0},\hat{1}, \emptyset, \times}^{\abs{\vect{x}}}\setminus
\set{\hat{0}^{\abs{\vect{x}}},\hat{1}^{\abs{\vect{x}}},
\emptyset^{\abs{\vect{x}}},
\times ^{\abs{\vect{x}}}
}
}$
and $\beta_{\vect{y}}\in \Complex$.
\end{itemize}
\end{claim}

That is, 
if $\abs{\vect{x}}$ is zero, then the state over $(\reg{R},\reg{G})$ is
 a tensor product of $\ket{\emptyset}$ by 
 Claims~\ref{claim:step3-2}~and~\ref{claim:step4}.

If $\abs{\vect{x}}$ is exactly one, i.e., 
$S=\set{i^*}$ for some $i^*$,
then, whatever unitary operator $W_h$ has acted on $\reg{R}_{i^*}$,
the resulting state
$\ket{\psi^{(4)}_{\vect{x}}}_{\reg{R}_S}\deq W_h\ket{\psi^{(3)}_{\vect{x}}}_{\reg{R}_S}$
is a quantum state over the register $\reg{R}_{i^*}$.
This state is obviously consistent over $S$,
since there is only one active party.

If $\abs{\vect{x}}$ is at least two
and $h$ equals $\abs{\vect{x}}$,
STEP 5 transforms
the state $\ket{\psi^{(3)}_{\vect{x}}}_{\reg{R}_S}$
into $\ket{\psi^{(4)}_{\vect{x}}}_{\reg{R}_S}$,
which is an inconsistent state over $S$.
Intuitively, 
the state $\ket{\psi^{(4)}_{\vect{x}}}_{\reg{R}_S}$
is a superposition
of only the basis states that
correspond to the situations  where
at least two active parties have different contents chosen from
$\set{\hat{0},\hat{1},\emptyset, \times}$
in their $\reg{R}_i$.
This fact exhibits a striking difference from the case of $\abs{\vect{x}}\le 1$.
We should emphasize
that 
the amplitude $\beta_{\vect{y}}$ vanishes
for each 
$\vect{y}\in 
\set{\hat{0}^{\abs{\vect{x}}},\hat{1}^{\abs{\vect{x}}},
\emptyset^{\abs{\vect{x}}},
\times ^{\abs{\vect{x}}}}$
\emph{only when} $h$ equals $\abs{\vect{x}}$.
In general, however, $h$ may not be equal to $\abs{\vect{x}}$,
in which case the amplitude $\beta_{\vect{y}}$ 
is nonzero for some
$\vect{y}\in 
\set{\hat{0}^{\abs{\vect{x}}},\hat{1}^{\abs{\vect{x}}},
\emptyset^{\abs{\vect{x}}},
\times ^{\abs{\vect{x}}}}$.

\begin{step}
Every active party measures 
$\reg{R}_i$
in the basis 
$\set{\ket{\hat{0}},\ket{\hat{1}}, \ket{\emptyset}, \ket{\times}}$
and 
obtains a two-bit classical outcome $r_i$.
\end{step}
Suppose that $\abs{\vect{x}}$ is at least one and that the quantum state over $\reg{R}_S$ just before STEP 6 is $\ket{\psi^{(4)}_{\vect{x}}}_{\reg{R}_S}$.
If $\abs{\vect{x}}$ is at least two
and $h$ equals $\abs{\vect{x}}$,
Claim~\ref{claim:step5} implies that
there are at least two distinct outcomes
among those obtained by active parties.
If $\abs{\vect{x}}$ equals one,
then there is obviously a single outcome
in the network.

Finally, suppose that
$\abs{\vect{x}}$ is zero.
Since there are no active parties, there are no outcomes
that would be obtained by the measurement.

The number of distinct outcomes is hence different between 
these three cases.

\begin{step}
All parties collaborate to
decide whether 
the number of distinct elements among
the outcomes $\set{r_i\colon i\in S}$
is at most one
by running the distributed deterministic algorithm 
provided in {Proposition~\ref{pr:ColorCounting}}.
If the number is at most one,
return $\true$; otherwise,
return $\false$.
\end{step}
To realize STEP~7,
it suffices to 
decide whether
the string $(r_1,\dots, r_n)$
is consistent over $S$,
where $r_j$ for $j\not\in S$ is set at an appropriate value
that is distinguishable from any possible outcome $r_i$ for $i\in S$
(technically, $r_i$ is a three-bit value).
For this purpose, 
we use the distributed algorithm $\CONSISTENCY$
given in 
Proposition~\ref{pr:ColorCounting}
(a special form of a more general statement~\cite{KraKriBer94InfoComp}).
The algorithm
yields the correct output if an upper bound of the diameter of the underlying graph is known. In our setting, $N$ can be used to upper-bound the diameter.
Every party can therefore decide whether
$T_1(\vect{x})$ is $\true$ or $\false$.

\begin{proposition}[Color Counting~\cite{KraKriBer94InfoComp}]
\label{pr:ColorCounting}
Let $G\deq (V,E)\in \calD_n$ be the 
the underlying graph
of an anonymous network
with a diameter upper-bounded by $\Delta$.
Let $S\subseteq [n]$ be the set of active parties
and
  let $C$ be the set of a constant number of colors.
Then, there is a deterministic algorithm $\COLORCOUNT$ that,
if the upper bound $\Delta$, 
 a color $c_i\in C$, and 
a bit $S_i$
indicating 
whether $i\in S$ or not are given at each party $i\in [n]$,
decides 
which is true among the following three cases
and informs every party of the decision$:$
\begin{center}
$(0):$ $\Abs{\bigcup_{i\in S}\set{c_i}}= 0$,\,\,\,\,
$(1):$ $\Abs{\bigcup_{i\in S}\set{c_i}}=1$,\,\,\,\,
$(2):$ $\Abs{\bigcup_{i\in S}\set{c_i}}\ge 2$\\
\end{center}
 in $O(\Delta)$ rounds with $O(\Delta \abs{E})$ classical bits of communication
$[$the case $(0)$ occurs if and only if $S$ is the empty set$]$.
As a special case, the algorithm  decides
whether  all active parties are assigned a certain single color
$[$called the ``$\consistent$'' case, corresponding to $(0)$ or $(1)$$]$ 
or at least two colors $[$called the ``$\inconsistent$'' case, corresponding to $(2)$$]$.
When the algorithm is used for this purpose, 
it is denoted by $\CONSISTENCY$.
\end{proposition}

The operations performed by each party $i$ 
in executing \Qhm
are described in Algorithm \Qhm shown below.
We should emphasize that
these operations
are independent of the index $i$
(recall that this is the requirement of computing on anonymous networks).
Indeed,
when executing \Qhm,
each party need not
tell its own registers from the other parties';
it just needs to distinguish between its local registers
and perform operators on them
in a way independent of its index
(but dependent on input).
Thus, the subscript $i$ of each register/variable
in Algorithm~\Qhm can safely be dropped
without introducing any ambiguity from the viewpoint of the party $i$.

\begin{algorithm}[t]
\small
\SetAlgoLined
\SetAlgoRefName{$\mbox{\boldmath \Qhm}$}
\caption{Every party $i$ performs the following operations.}
\label{Qhm}
\SetKw{Notation}{Notation}
\KwIn{a classical variable $x_{i}$, and $N\in \Natural$.}
\KwOut{$\true$ or $\false$.}
\BlankLine
\Notation{
Let $\reg{R}_i$, $\reg{Y}_i$, and $\reg{G}_i$
be $($a set of$\,)$ quantum registers, and
let $y_i$, $r_i$, and $\bar{r}_i$ be classical variables.
Assume that 
 $\ket{\hat{0}},
\ket{\hat{1}},
\ket{\emptyset},
\ket{\times}$
are
$\ket{00}, \ket{01}, \ket{10}, \ket{11}$,
respectively.
}

\BlankLine
\Begin{
\nlset{STEP1}\lIf{$x_{i}=1$}{prepare $\frac{1}{\sqrt{2}}(\ket{\hat{0}}+\ket{\hat{1}})$
in $\reg{R}_i$}
\lElse{prepare $\ket{\emptyset}$ in $\reg{R}_i$}.

\Begin{
\nlset{STEP2}run $\QCONSISTENCY$
on $\reg{R}_i$
with $(S_i,\Delta)\leftarrow (x_i,N)$ 
to output ($\reg{R}_i,\reg{Y}_{i},\reg{G}_i)$

($\reg{G}_i$ is the party's share of $\reg{G}$, so that $\reg{G}=(\reg{G}_1,\dots, \reg{G}_n)$)\;

\nlset{STEP3}measure $\reg{Y}_{i}$
        in the 
        $\{ \ket{\mbox{``$\consistent$''}},
        \ket{\mbox{``$\inconsistent$''}} \}$ basis to obtain an outcome $y_i$\;

\lIf{$y_i$ is ``$\inconsistent$''}{
       \Return{$\false$}.}
}

\nlset{STEP4}\Begin{
run $\GHZSCALEDOWN$ on $(\reg{R}_i,\reg{G}_i)$ with $(h,m)$\;
\lIf{the output is ``$\true$''}{\Return{$\true$}}.
}

\nlset{STEP5}\lIf{$x_{i}=1$}{apply $W_h$ to $\reg{R}_i$}\;

\nlset{STEP6}\lIf{$x_{i}=1$}{
measure $\reg{R}_i$ in the computational basis
to obtain an outcome $r_i\in \set{00,01,10,11}$\;
}

\lElse{
set $r_i$ at ``$\emptyset$''\;
}

\nlset{STEP7}run $\CONSISTENCY$ on $r_i$ with 
$(S_i,\Delta)\leftarrow (x_i,N)$ 
to output $\bar{r}_i$\;

\lIf{$\bar{r}_i$ is ``$\consistent$''}{\Return{$\true$}}
\lElse{\Return{$\false$}.}
}

\end{algorithm}

\section{Proof of  Lemma~\ref{lm:Q{h,m}}}
\label{sec:proofs}
This section first provides all the proofs of the propositions and claims
in Sec.~\ref{sec:keysub}
(except those appearing in previous works).
Then it proves Lemma~\ref{lm:Q{h,m}}.

\subsection{Proofs of Claims and Propositions  in Sec.~\ref{sec:SVAlgorithm}}
\label{subsec.Q{h,m}}
All proofs except the one for Claim~\ref{claim:step4}
essentially follow from the ideas in Refs.~\cite{KraKriBer94InfoComp,TanKobMat12TOCT}.
Readers familiar with these references can safely skip
the proofs.
The proof for Claim~\ref{claim:step4} is based on Ref.~\cite{TanKobMat12TOCT} for the special case, but it also discusses the other cases that the reference does not deal with.

Proposition~\ref{pr:ColorCounting} is simply an adaptation of Theorem~1 in 
Ref.~\cite{KraKriBer94InfoComp}.
For completeness, we provide its proof.

\begin{proofof}{Proposition~\ref{pr:ColorCounting}}
Suppose that every party $i$ prepares a variable $Y$.
Each active party initializes $Y$ to $\set{c_i}$,
while
each inactive party initializes $Y$ to ``$\emptyset$'',
representing
the empty set.
Every party then sends a copy of $Y$ via every out-port
while keeping  a copy of $Y$
and receives a message as a variable $Z_{k}$  via every in-port $k$.
The party then updates $Y$ by setting 
it at the union of $Y$ and $Z_k$ over all $k$.
Every party repeats the above sending/receiving
at most $\Delta$ times.
It is easy to see that
every party's final  $Y$ is 
the set of all colors in active parties' possession.
Since each message is of constant size, the bit complexity is $O(\Delta \abs{E})$.
\end{proofof}

As a corollary of Proposition~\ref{pr:ColorCounting},
we obtain Proposition~\ref{pr:T_{0}}.

\begin{proofof}{Proposition~\ref{pr:T_{0}}}
Let $C:=\set{1}$ and $S:= \set{i\in [n] \colon x_{i}=1}$.
Every party runs $\COLORCOUNT$ in 
Proposition~\ref{pr:ColorCounting}
and decides which of the two cases, (0) or (1), holds
(event (2) never occurs since $\abs{C}=1$).
\end{proofof}

We henceforth provide proofs of the claims in order.

\begin{proofof}{Claim~\ref{claim:step2}}
The claimed quantum algorithm 
essentially follows from
``\CONSISTENCY$_d$" on page~21 of
Ref.~\cite{TanKobMat12TOCT},
where $n$ is replaced with $\Delta+1$ and
the binary operation ``$\circ$''
is interpreted as the union operation over sets in our case.
Note that ``\CONSISTENCY$_d$" is a simple quantization
of the classical algorithm in Proposition~\ref{pr:ColorCounting}.

The ancillary quantum registers left at \emph{each} party
after \CONSISTENCY$_d$
runs in $\Delta$ rounds are
(1) the master copy $\reg{Y}_{0}^{(t)}$ whose copies have been sent out to neighbors in each round $t\in [\Delta]$,
(2) all registers $(
\reg{Z}_{1}^{(t)},\dots, \reg{Z}_{d^{\text{in}}}^{(t)}
)$
received in each round $t\in [\Delta]$,
and (3) the final output register $\reg{Y}_{0}^{(\Delta+1)}$.
Let $\reg{Y}$ be the collection $(\reg{Y}_1,\ldots, \reg{Y}_n)$,
where $\reg{Y}_i$ for $i\in [n]$ denotes the register $\reg{Y}_{0}^{(\Delta+1)}$ 
in the possession of party $i$,
and
let $\reg{G}$ be all the ancillary registers in the network except $\reg{Y}$,
namely, 
the collection of 
$(\reg{Y}_{0}^{(1)},\dots, \reg{Y}_{0}^{(\Delta)})$, 
$(
\reg{Z}_{1}^{(1)},\dots, \reg{Z}_{d^{\text{in}}}^{(1)}
),\ldots,
(
\reg{Z}_{1}^{(\Delta)},\dots, \reg{Z}_{d^{\text{in}}}^{(\Delta)}
)$ over all parties $i$.
It is easy to see that 
the whole action of the quantum algorithm can be written
as the operator $\Phi_{\Delta}$ in the claim.

For the complexity,
$\CONSISTENCY_d$ requires $O(\Delta)$ rounds.
In addition, $\CONSISTENCY_d$ communicates $O(\Delta \abs{E})$ qubits,
since a single register representing a subset of $\set{\hat{0},\hat{1},\emptyset, \times}$
is sent through each communication link at each round.
\end{proofof}

\begin{proofof}{Claim~\ref{claim:step3-1}}
The claim is very easy. We omit the proof.
\end{proofof}

\begin{proofof}{Claim~\ref{claim:step3-2}}
Since the state 
$\ket{\psi^{\text{(2)}}_{\vect{x}}}_{(\reg{R},\reg{G})}$
is a superposition of $
\ket{\vect{z}}_{\reg{R}}\otimes \ket{g_{\Delta}(\vect{z})}_{\reg{G}}$
with $\vect{z}\in A$,
the contents of the registers $\reg{R}_i$ of
all active parties $i$ are identical
by the definition of $A$.
More precisely, they are either all ``$\hat{0}$'' or all ``$\hat{1}$'' with probability
$1/2$.
Meanwhile, the contents of $\reg{R}_j$ of all inactive parties $j$
are ``$\emptyset$''.
The state $\ket{\psi^{\text{(2)}}_{\vect{x}}}_{(\reg{R},\reg{G})}$ is hence 
of the form
\[
\ket{\psi^{\text{(2)}}_{\vect{x}}}_{(\reg{R},\reg{G})}=
\displaystyle{
\frac{1}{\sqrt{2}}
\Oset{
\ket{\hat{0}}^{\otimes \abs{\vect{x}}}
\otimes
\Ket{g(\vect{z}_{\hat{0}})}
+
\ket{\hat{1}}^{\otimes \abs{\vect{x}}}
\otimes
\Ket{g(\vect{z}_{\hat{1}})}
}
_{(\reg{R}_{S},\reg{G})}
\otimes
}
\Oset{
\ket{\emptyset}^{\otimes (n-\abs{\vect{x}})}
}_
{\reg{R}_{\overline{S}}}.
\]
This fact and items 1 to 3 of the following Claim~\ref{cl:ContentsOfG}
imply 
\begin{equation}
\ket{\psi^{\text{(2)}}_{\vect{x}}}_{(\reg{R},\reg{G})}=
\frac{1}{\sqrt{2}}
\Oset{\ket{\hat{0}}^{\otimes p}+\ket{\hat{1}}^{\otimes p}}
_{(\reg{R}_S,\reg{G}')}
\otimes 
\Oset{\ket{\emptyset}^{\otimes q}}
_{(\reg{R}_{\bar{S}},\reg{G}'')},
\end{equation}
where 
$\reg{G}''$ is the set of (unentangled) registers in $\reg{G}$ whose contents are ``$\emptyset$'', and $\reg{G}'$ is the remaining registers in $\reg{G}$.
It is obvious that  $p+q\ge \abs{S}+\abs{\bar{S}}=n$.
It also holds that $p\ge \abs{S}$,
since the content of $\reg{R}_i$ for $i\in S$ is either ``$\hat{0}$'' and ``$\hat{1}$''.
In particular, when $\abs{S}$ is $0$, 
item 4 of Claim~\ref{cl:ContentsOfG}
implies that $p$ equals 0.
This completes the proof of Claim~\ref{claim:step3-2}.
\end{proofof}

\begin{claim}
\label{cl:ContentsOfG}
For a fixed set $S$ of active parties,
suppose that 
the state of $(\reg{R},\reg{G})$ is
$\ket{\psi^{\text{\rm (2)}}_{\vect{x}}}_{(\reg{R},\reg{G})}$.
Then the following hold.
\begin{enumerate}
\item 
If the content of $\reg{R}_i$ is ``$\hat{0}$'' for every $i\in S$,
then the content of every register in $\reg{G}$ is 
either ``$\hat{0}$'' or ``$\emptyset$''.

\item 
If the content of $\reg{R}_i$  is ``$\hat{1}$'' for every $i\in S$,
then the content of every register in $\reg{G}$ is 
either ``$\hat{1}$'' or ``$\emptyset$''.

\item
If a register in $\reg{G}$ has the content ``$\emptyset$'',
then the register is unentangled.

\item If there are no active parties, then 
the content of every register in $(\reg{R},\reg{G})$ is ``$\emptyset$''.
\end{enumerate}
\end{claim}
\begin{proof}
Claim~\ref{cl:ContentsOfG}
essentially follows from the fact that
the operator $\Phi_{\Delta}$, i.e., $\QCONSISTENCY$,
consists of idempotent operations, copy operations via CNOT,
and register exchanges. We thus omit the proof.
This claim is (implictly) observed 
in Sec.~4 of Ref.~\cite{TanKobMat12TOCT}.
\end{proof}

\begin{proofof}{Claim~\ref{claim:step4}}
The algorithm $\GHZSCALEDOWN$
is 
a modification of the algorithm sketched 
in Sec.~4 of Ref.~\cite{TanKobMat12TOCT},
which is obtained by
combining standard quantum and  classical techniques.
That is, every party $i$
performs the following procedure
on $(\reg{R},\reg{G})$,
which is supposed to be in the state
$\ket{\psi^{(2)}_{\vect{x}}}_{(\reg{R},\reg{G})}$.
Let $\ket{\hat{+}}\deq (\ket{\hat{0}}+\ket{\hat{1}})/\sqrt{2}$ and
$\ket{\hat{-}}\deq (\ket{\hat{0}}-\ket{\hat{1}})/{\sqrt{2}}$.
\begin{enumerate}
\item[(i)]
Measure every local register in $\reg{G}$
in the basis
$\set{\ket{\hat{+}},\ket{\hat{-}},\ket{\emptyset},\ket{\times}}$.
\item[(ii)]
Locally count the number $s_i \pmod 2$ of the measurement outcomes ``$\hat{-}$''
that party $i$ has obtained.
\item[(iii)]
Attempt to compute the sum 
of $s_i$ $\pmod 2$ over all parties
by invoking the algorithm in Fact~\ref{th:view}
with $s_i$ and $m$. Let  $\chi$ be the output of the algorithm.
\item [(iv)] Decide with $\CONSISTENCY$ given in 
Proposition~\ref{pr:ColorCounting} whether 
the string induced by $\chi$'s  is consistent over all parties.
If either the result is ``$\inconsistent$'' or $\chi$ is not a non-negative integer, 
then output ``$\true$'' and halt.
\item [(v)]
If the party $i$ is active (i.e., $i\in S$), 
the output $\chi \pmod 2$ equals one,
and $h$ is at least $1$,
then perform
    $R({\frac{\pi}{h}})\deq
    \Cset{
      \begin{smallmatrix}
        1 & 0\\
        0 & \e ^{\im\frac{\pi}{h}}
      \end{smallmatrix}
     }
   $ 
on 
the subspace spanned by $\set{\ket{\hat{0}},\ket{\hat{1}}}$
on $\reg{R}_i$ and output  $\reg{R}_i$.
\end{enumerate}
The reason Step (iv) is needed is as follows:
For $m\neq n$, the algorithm involved in Step (iii) may not work correctly,
and some party could output $\chi$ that is different from other parties'
$\chi$ and/or is even a non-integer value
(actually, the algorithm never outputs negative values,
which follows from the details of the algorithm, which we do not touch on in this paper). Step (iv) is for the purpose of
detecting that the guess at $m$ is wrong and
precluding this wrong guess from leading to the wrong final decision
when the number of active parties is at most one.

First assume that $(h,m)$ equals $(\abs{\vect{x}},n)$ and $\abs{\vect{x}}$ is at least one.
In this case, 
Steps (i)-(v) 
transform the state 
$\ket{\psi^{\text{(2)}}_{\vect{x}}}$
to
$
\ket{\psi^{(3)}_{\vect{x}}}_
{\reg{R}_S}=
(
\ket{\hat{0}}^{\otimes \abs{\vect{x}}}
+
\ket{\hat{1}}^{\otimes \abs{\vect{x}}}
)_
{\reg{R}_S}/{\sqrt{2}},
$
by following Sec.~4 in Ref.~\cite{TanKobMat12TOCT}.

Next suppose that $\abs{\vect{x}}$ is zero.
For any pair $(h,m)$,
every party can still perform Steps (i) and (ii)
since these steps are independent of $(h,m)$.
Every party then performs Step~(iii).
Note that this is possible even when $m$ is not equal to $n$
(see the description just after Fact~\ref{th:view}).
If $m$ is not equal to $n$, however,
the value $\chi$ obtained in Step (iii) may be wrong
or the string induced by $\chi$'s could be inconsistent over all parties.%
\footnote{
When $\abs{\vect{x}}$ is zero, 
$s_i$ is  zero for all  $i$.
This actually implies that 
$\chi$ is zero for all parties
regardless of $m$,
which follows from the details found in 
Ref.~\cite{YamKam96IEEETPDS-1}.
However, we consider the possibility of halting at Step (iv)
to avoid getting into the details of the algorithm in Fact~\ref{th:view}.
}
Since $\CONSISTENCY$ can run 
and make a common decision for each party
by setting $\Delta$ at $N$ (Proposition~\ref{pr:ColorCounting}),
either every party halts with output ``$\true$'' at Step~(iv)
or every party proceeds to Step~(v)
with a non-negative integer $\chi$ that agrees with any other party's $\chi$.
In the latter case,
no operations
are performed 
in Step~(v)
since no active parties exist (i.e., $\abs{\vect{x}}=\abs{S}=0$).
Thus, the procedure effectively applies the identity operator to 
$\ket{\psi^{(2)}_{\vect{x}}}_{(\reg{R},\reg{G})}$, i.e., a tensor product of $\ket{\emptyset}$.

Finally, suppose that $(h,m)$ is not equal to $(\abs{\vect{x}},n)$
and $\abs{\vect{x}}$ is at least one.
For any pair $(h,m)$,
every party can still perform Steps (i), (ii), and (iii).
As in the case of $\abs{\vect{x}}=0$,
either every party halts with output ``$\true$'' at Step~(iv)
or every party proceeds to Step~(v)
with a non-negative integer $\chi$ that agrees with any other party's $\chi$.
Suppose that it proceeds to Step~(v).
If $h$ is at least $1$ and the value $\chi \pmod 2$
happens to be one, then every active party applies 
$R({\frac{\pi}{h}})$.
This effectively multiplies
the state $\ket{\hat{1}}^{\otimes \abs{\vect{x}}}_{\reg{R}_S}$
by a factor
 $e^{i\frac{\pi}{h}\abs{\vect{x}}}$;
namely, it transforms
$\ket{\psi^{(2)}_{\vect{x}}}_{\reg{R}_S}$
to 
$\ket{\tilde{\psi}^{(3)}_{\vect{x}}}_{\reg{R}_S}
\deq
(1/{\sqrt{2}})
(\ket{\hat{0}}^{\otimes \abs{\vect{x}}}\pm e^{i\frac{\pi}{h}\abs{\vect{x}}}\ket{\hat{1}}^{\otimes \abs{\vect{x}}})_{\reg{R}_S}$.
If $h$ or $\chi \pmod 2$ equals zero, then
Step~(iv) performs no operations, and thus the entire procedure
effectively transforms
$\ket{\psi^{(2)}_{\vect{x}}}_{\reg{R}_S}$
to 
$\ket{\tilde{\psi}^{(3)}_{\vect{x}}}_{\reg{R}_S}
\deq 
(1/{\sqrt{2}})
(\ket{\hat{0}}^{\otimes \abs{\vect{x}}}\pm \ket{\hat{1}}^{\otimes \abs{\vect{x}}})_{\reg{R}_S}$.

For the communication cost, observe that only Steps (iii) and (iv) involve (classical) communication.
Step~(iii) just runs the algorithm provided in Fact~\ref{th:view} with $m$ instead of (unknown) $n$. In this case, the algorithm
runs in $O(m)$ rounds and communicates 
$O(\max\set{p(m),p(n)})$ classical bits, 
as described just after Fact~\ref{th:view}.
The cost of Step~(iv) is shown in 
Proposition~\ref{pr:ColorCounting}
and dominated by the cost of Step~(iii).
\end{proofof}

\begin{proofof}{Claim~\ref{claim:step5}}
If $\abs{\vect{x}}$ is zero, there exist no active parties.
STEP~5 is hence effectively skipped.
If $\abs{\vect{x}}$ is one, the register $\reg{R}_S$ consists of two qubits. 
Since $\set{\ket{\hat{0}},\ket{\hat{1}}, \ket{\emptyset}, \ket{\times}}$
is an orthonormal basis of $\Complex^{4}$,
the statement follows.
Now suppose that $\abs{\vect{x}}$ is at least two.
In addition, assume that $h$ equals $\abs{\vect{x}}$.
Recall
that $(\ket{\hat{0}}, \ket{\hat{1}}, \ket{\emptyset},\ket{\times})$ denotes 
$(\ket{00},\ket{01},\ket{10},\ket{11})$.
We thus have $\ket{\psi^{(3)}_{\vect{x}}}_
{\reg{R}_S}=
(1/\sqrt{2})
 (\ket{00}^{\otimes \abs{\vect{x}}}+\ket{{01}}^{\otimes \abs{\vect{x}}})
 $.
Every active party
applies 
the unitary operator $W_h$
to its $\reg{R}_i$,
where $W_h$ is either $ I_2\otimes U_h$ for even $h$ or $V_h\cdot \CNOT_{2\rightarrow 1}$
for odd $h$.
Lemmas 3.3 and 3.5 in Ref.~\cite{TanKobMat12TOCT}
imply that
$\ket{\psi^{(3)}_{\vect{x}}}_
{\reg{R}_S}$ is transformed into
an inconsistent state over $S$ whenever $h\ge 2$:
$
\ket{\psi^{(3)}_{\vect{x}}}_{\reg{R}_S}
\mapsto 
\ket{\psi^{(4)}_{\vect{x}}}_{\reg{R}_S}
\deq
\sum_{\vect{y}\in 
B
}
\beta_{\vect{y}}\ket{\vect{y}}_{\reg{R}_S}
$
for certain amplitudes $\set{\beta_\vect{y}}$.
\end{proofof}

\subsection{Proof of Lemma~\ref{lm:Q{h,m}}}
We first show that, for any $(h,m)$,
STEPs~1 and 2 can run and yield
$\ket{\psi^{(1)}_{\vect{x}}}_{(\reg{R},\reg{Y},\reg{G})}$.
Since STEP~1 is independent of $(h,m)$, it always yields
$\ket{\psi^{(0)}_{\vect{x}}}_{\reg{R}}$.
STEP~2  runs the algorithm provided in Claim~\ref{claim:step2},
which only requires an upper bound $\Delta$ on the diameter 
of the underlying graph $G$.
By setting $\Delta$ at $N$, Claim~\ref{claim:step2}
implies that STEP~2  always yields $\ket{\psi^{(1)}_{\vect{x}}}_{(\reg{R},\reg{Y},\reg{G})}$.

For STEP~3 and latter steps, we consider three cases separately:
(A) $\abs{\vect{x}}=0$;
(B) $\abs{\vect{x}}\ge 1$ and $(h,m)=(\abs{\vect{x}},n)$;
(C) $\abs{\vect{x}}\ge 1$ and $(h,m)\neq (\abs{\vect{x}},n)$.

Suppose case (A): $\abs{\vect{x}}=0$.
STEP~3 can be performed, since it consists of only a measurement that is independent of $(h,m)$.
Moreover, every party can obtain the same outcome after STEP~3,
since $\ket{\psi^{(1)}_{\vect{x}}}_{(\reg{R},\reg{Y},\reg{G})}$
is of the form shown in Eq.~(\ref{eq:state1}).
The common outcome is ``$\consistent$'' with certainty
by Claim~\ref{claim:step3-1}.
Thus, the state resulting from the measurement is always
$\ket{\psi^{(2)}_{\vect{x}}}_{(\reg{R},\reg{G})}$, 
which is a tensor product of $\ket{\emptyset}$ by Claim~\ref{claim:step3-2}.
STEP~4 consists of a distributed algorithm $\GHZSCALEDOWN$, 
which  
either returns ``$\true$'' or keeps the tensor product in 
$(\reg{R},\reg{G})$,
as asserted by Claim~\ref{claim:step4}.
STEPs~5 and 6 are effectively skipped,
since there are no active parties. 
Hence, there are no measurement outcomes
in the network,
which STEP~6 would yield if it were not skipped. This forces STEP~7 to determine 
by using Proposition~\ref{pr:ColorCounting}
that the number of distinct elements among the outcomes is zero with certainty,  which leads to returning ``$\true$''.
Note that STEP~7 
consists of the distributed algorithm $\CONSISTENCY$
provided in Proposition~\ref{pr:ColorCounting},
which makes
the decision of ``$\true$'' at every party
for any $(h,m)$.
Therefore, items 1 through 3 in Lemma~\ref{lm:Q{h,m}} hold when $\abs{\vect{x}}$ is zero.

Next, suppose case (B): $\abs{\vect{x}}\ge 1$ and 
$(h,m)=(\abs{\vect{x}},n)$.
Claim~\ref{claim:step3-1} implies that, if $\abs{\vect{x}}$ is one, 
then each party obtains with certainty the outcome ``$\consistent$''
of the measurement made in STEP~3.
Note that, with the same argument as in case (A),
the outcomes of all parties agree.
STEP~3 can thus measure ``$\inconsistent$'' and output ``$\false$'' 
only if $\abs{\vect{x}}$
is at least two. 
Hence, whenever STEP~3 outputs ``$\false$'',  this output agrees with $T_1(\vect{x})$.
Assume henceforth that the outcome is ``$\consistent$''.
The resulting state is then
$\ket{\psi^{(2)}_{\vect{x}}}_{(\reg{R},\reg{G})}$, 
which is of the form in Eq.~(\ref{eq:state2})  in Claim~\ref{claim:step3-2}.
STEP~4 then transforms
$\ket{\psi^{(2)}_{\vect{x}}}_{(\reg{R},\reg{G})}$
to $\ket{\psi^{(3)}_{\vect{x}}}_{\reg{R}_S}$
with certainty as asserted by Claim~\ref{claim:step4}.
STEP~5 further transforms
$\ket{\psi^{(3)}_{\vect{x}}}_{\reg{R}_S}$
to $\ket{\psi^{(4)}_{\vect{x}}}_{\reg{R}_S}$
as implied by Claim~\ref{claim:step5}.
If $\abs{\vect{x}}(=\abs{S})$ is one,
then $\ket{\psi^{(4)}_{\vect{x}}}_{\reg{R}_S}$
is exactly the state of $\reg{R}_i$ of the only active party $i$,
and thus only a single outcome of the measurement 
is obtained in the whole network in STEP~6.
In this case, STEP~7  returns ``$\true$'', which matches $T_1(\vect{x})$.
If $\abs{\vect{x}}$ is at least two,
then the state $\ket{\psi^{(4)}_{\vect{x}}}_{\reg{R}_S}$
is inconsistent over $S$ by Claim~\ref{claim:step5},
so that the string induced by the measurement outcomes obtained in STEP~6
is also inconsistent over $S$.
Thus, there are two or more distinct outcomes in the whole network, and STEP~7 returns ``$\false$'', matching $T_1(\vect{x})$.
Therefore, items 1 through 3 in Lemma~\ref{lm:Q{h,m}} hold
if $\abs{\vect{x}}\ge 1$ and 
$(h,m)=(\abs{\vect{x}},n)$.

Finally, suppose case (C): $\abs{\vect{x}}\ge 1$ and $(h,m)\neq (\abs{\vect{x}},n)$. It suffices to show 
that items 1 and 3 in Lemma~\ref{lm:Q{h,m}}
 hold in this case;
namely,
that STEPs~3 through 7 can be performed,
\Qhm returns a common decision (i.e., $\true$ or $\false$) to every party
for any $(h,m)$, and
\Qhm always returns ``$\true$'' if $\abs{\vect{x}}$ is one.
As described in case (A),
every party obtains the same outcome of the measurement
for any $(h,m)$ in STEP~3.
With the same argument used in case (B), 
if $\abs{\vect{x}}$ is one, then 
Claim~\ref{claim:step3-1} implies that 
STEP~3 never returns ``$\false$'' 
and every party proceeds to STEP~4
with the state 
$\ket{\psi^{(2)}_{\vect{x}}}_{(\reg{R},\reg{G})}$.
Claim~\ref{claim:step4} then implies that,
for any $\abs{\vect{x}}\ge 1$ and
for any $(h,m)\neq (\abs{\vect{x}},n)$,
STEP~4 either returns ``$\true$'' to every party or transforms
$\ket{\psi^{(2)}_{\vect{x}}}_{(\reg{R},\reg{G})}$
to $\ket{\tilde{\psi}^{(3)}_{\vect{x}}}_{\reg{R}_S,}$.
Assume the latter case.
In STEP~5, every active party applies the local unitary operator 
$W_h$ to its share of 
 $\ket{\tilde{\psi}^{(3)}_{\vect{x}}}_{\reg{R}_S,}$.
This is possible for any $(h,m)$, since $W_h$ is defined for every possible $h$.
STEP~6 can obviously be performed, since it consists of a measurement
that is independent of $(h,m)$.
For any $(h,m)$,
STEP~7 
makes
a common decision at every party 
as stated in Proposition~\ref{pr:ColorCounting}.
If $\abs{\vect{x}}$ is one,
the register $\reg{R}_S$ is exactly $\reg{R}_i$  of the only active party $i$.
Therefore, whatever state in $\reg{R}_S$ results from STEP~5,
STEP~6 yields only a single outcome in the whole network.
STEP~7 thus returns ``$\true$''.
This shows that  items 1 and 3 in Lemma~\ref{lm:Q{h,m}}
 hold.

To bound the complexity, observe that
all the communication 
performed by \Qhm is devoted to 
STEPs~2, 4, and 7.
The complexities of these steps are shown in 
Claims~\ref{claim:step2} and \ref{claim:step4} and Proposition~\ref{pr:ColorCounting} with $N$ as $\Delta$.
Summing them up shows that \Qhm runs
in $O(N)$ rounds and communicates $O(N^3)$ qubits
and $O(p(N))\subseteq \tilde{O}(N^6)$ classical bits.

\section{Applications}
This section provides  some applications of the solitude verification algorithm.
\subsection{Zero-Error Leader Election (Proof of Corollary~\ref{cr:LE})}
The algorithm in Theorem~\ref{th:SV}, called QSV,
leads to a simple zero-error algorithm for the leader election problem
(a pseudo-code is given as Algorithm~\ref{ZQLE}).
This application is somewhat standard, but we will sketch how it works
for completeness. 

For every $s\in [2..N]$,
every party $i$ sets 
$x_{i}^{(s)}$ to a random bit,
which is $1$ with probability $1/s$
or $0$ with probability $1-1/s$.
The party then performs QSV
with  $x_i^{(s)}$
and $N$ over all $s$ in parallel.
Let QSV$[\vect{x}^{(s)},N]$ be the (common) output of QSV 
that every party obtains,
where
$\vect{x}^{(s)}\deq (x_{1}^{(s)},\dots, x_{n}^{(s)})$.
If there exists at least one $s$ such that QSV$[\vect{x}^{(s)},N]$
is ``$\true$'',
then every party $i$ outputs $z_{i}\deq x_{i}^{(s_{\max})}$,
where $s_{\max}$ is the maximum of $s$ 
such that QSV$[\vect{x}^{(s)},N]$ is ``$\true$'';
Otherwise,
it gives up.
The party with $z_{i}=1$ is elected as a unique leader.
Note that this elects a unique leader without error whenever  
$s_{\max}$ exists.
The probability of successfully electing a unique leader
is at least some constant,
since for $s=n$,
the probability that there is exactly a single
$i\in [n]$ 
with $x_{i}^{(s)}=1$ 
is $\binom{n}{1}\frac{1}{n}(1-\frac{1}{n})^{{n-1}}> 1/e$.
By the standard argument,
this probability can be amplified to a constant arbitrarily close to one
by simply repeating QZLE sufficiently many but constant times.
Since all communication in QZLE
is devoted to QSV, which runs for all $s\in [2..N]$
in parallel, the overall round complexity is still $O(N)$
and the overall bit complexity is 
$\tilde{O}(N^8)\times (N-1)=\tilde{O}(N^9)$.

\begin{figure}
\small
\begin{algorithm}[H]
\SetAlgoLined
\SetKw{Notation}{Notation:}
\KwIn{an upper bound $N$ on the number of parties.}
\KwOut{$z_{i}\in \set{0,1, \mathsf{give\mbox{-}up}}$.}
\BlankLine
\Notation{
\emph{Let $\calP_{s}$ be the distribution over $\set{0,1}$ for which
$\Pr _{Z\in_{\calP_{s}} \set{0,1}}[Z=1]=1/s$.}
}
\BlankLine
\Begin{
\ForAll{$s\in [2..N]$}{
perform in parallel Algorithm~QSV with input $(x_{i}^{(s)},N)$  to obtain output $y_i(s)$,\\
where $x_{i}^{(s)}\in _{ \calP_{s}}\set{0,1}$.
}

\uIf{$y_i(s)=\true$ for some $s$}{
\Return{$z_{i}\deq x_{i}^{(s_{\max})}$}, where $s_{\max}\deq\max \set{s\colon y_{s}=\true}$\;
}
\Else{\Return{$z_{i}\deq \mathsf{give\mbox{-}up}$}.}
}
\SetAlgoRefName{{ZQLE}}
\caption{Every party $i$ performs the following operations.} \label{ZQLE}
\end{algorithm}
\end{figure}

\subsection{Computing General Symmetric Functions (Proof Sketch of Theorem~\protect\ref{th:symmetric})}
The simple idea in the formal proof
is likely to be hidden
under complicated notations.
We thus only sketch the proof and relegate its formal description
to Appendix~\ref{appdx:symmetric}.

Recall that if there are at least two active parties,
then during the execution of Algorithm~QSV (on page~\pageref{QSV}),
Algorithm \Qhm (on page \pageref{Qhm}) for some $(h,m)$  outputs ``$\false$''.
Let $(h^*,m^*)$ be the lexicographically smallest 
pair among the pairs  $(h,m)$  for which \Qhm outputs ``$\false$''
on $\vect{x}$ and $N$.
Note that the decision ``$\false$'' must have been made at either STEP~3 or STEP~7 of $Q_{h^*,m^*}$.
If STEP~3 outputs ``$\false$'',
we insert the new step where every party measures its $\reg{R}_i$.
Since the state over all $\reg{R}_i$'s is an inconsistent state over the set $S$
of active parties in this case,
the string induced by the set of all the measurement outcomes 
is inconsistent over $S$.
The string thus partitions the set of active parties
into equivalence classes
naturally
defined by the outcomes.
Similarly, 
if STEP~7 outputs ``$\false$'',
then for the measurement outcomes $r_i\ (i\in [n])$ defined in $Q_{h^*,m^*}$,
the string $r_1\ldots r_n$ 
is inconsistent over the set $S$ of active parties.
The string thus partitions 
the set of active parties 
into equivalence classes
naturally
defined by $r_i$'s.
By repeating this process recursively,
the active parties will eventually be partitioned into equivalence classes
$(V_1,\dots, V_l)$
such that at least one of them is a singleton.
This can be verified as follows:
For each equivalence class $V_j$,
run QSV
with the members of $V_j$
as active parties.
If two or more singleton classes are found,
then all parties agree on one of the singleton classes
in an arbitrary way.
The parties then decide that the unique member of the class be a leader.
It is not difficult to show that
once a unique leader is elected,
the leader can compute $\abs{\vect{x}}$ in $O(N)$ rounds
with a polynomially bounded  bit complexity,
which is more formally stated as Claim~\ref{cl:eval|x|}.
\begin{claim}
\label{cl:eval|x|}
Suppose that there are $n$ parties on an anonymous (classical) network
with any underlying graph $G\deq (V,E)$ in $\calD_n$
in which an upper bound $N$ on $n$ is given as  global information.
Suppose further that 
each party $i$ in the network has a variable $S_i$
such that 
$S_l=$``$\mathsf{leader}$''  for a certain  $l\in [n]$
and $S_i=$``$\mathsf{follower}$'' for all $i\in [n]\setminus\set{l}$.
If every party $i$ is given $x_i\in \set{0,1}$,
then every party can compute 
$\abs{\vect{x}}$ 
in $O(N)$ rounds with the bit complexity ${O}(N\abs{V}^2\abs{E}\log \abs{V})$.
\end{claim}
The parties can thus tell the value of $f(\vect{x})$ for any fixed 
$f\in\calS_{n}(k)$.
Observe that such a singleton class 
appears within $\ceil{\log_2 \abs{\vect{x}}}$ levels of recursion.
Then, it suffices 
for the following reason
to continue the process
up to $\ceil{\log_2 k}$-th recursion level:
If there are no singleton sets 
at $\ceil{\log_2 k}$-th recursion,
then $\abs{\vect{x}}$ must be larger than $k$, and thus the parties
can determine the value of $f(\vect{x})$;
otherwise, the parties can compute $f(\vect{x})$ as we have already shown.
The total number of rounds is thus $O(N\log (\max\set{k,2}))$.

\begin{remark}
\label{rm:alternative}
For readers familiar with Algorithm II in Ref.~\emph{\cite{TanKobMat05STACS}}
(which works even for any strongly connected directed graph),
an alternative algorithm 
for computing $f\in \calS(k)$
can be considered as follows$:$
First start Algorithm II and stop after the first $\ceil{\log_2 k}$ phases have finished.
Then, verify with \emph{Algorithm~QSV} that a unique leader is elected.
If this is the case, the leader can compute $\abs{\vect{x}}$ 
as in {Claim~\ref{cl:eval|x|}}.
If it fails, then this implies that $\abs{\vect{x}}$ is more than
$2^{\ceil{\log k}}\ge k$ and thus determines the value of $f(\vect{x})$.
Since each phase consists of $O(N)$ rounds,
the whole algorithm runs in $O(N\log (\max\set{k,2}))$ rounds
$($with a polynomially bounded bit complexity$)$.
\end{remark}
\section*{Acknowledgments}
The author is grateful to anonymous referees for various helpful comments for improving the presentation
and suggesting the alternative approach mentioned in Remark~\ref{rm:alternative}.
The author is also grateful to the ELC project (Grant-in-Aid for Scientific Research on Innovative Areas No. 24106009 of the MEXT in Japan) for encouraging the research presented in this paper.
\bibliographystyle{alpha}

\newcommand{\etalchar}[1]{$^{#1}$}

\appendix
\section*{Appendix}

\label{appdx:symmetric}
\begin{proofof}{Claim~\ref{cl:eval|x|}}
To compute the value $\abs{\vect{x}}$, the leader first assigns a unique indetifier
$\mathsf{id}_i$
to each party $i$ and then 
every party collects all pairs $(\mathsf{id}_i,x_i)$
for $i\in [n]$ by using Proposition~\ref{pr:ColorCounting},
from which every party can compute the value $\abs{\vect{x}}$ locally.

To assign unique identifiers,
the leader first sends a message ``$j$'' of 
$O(\log \abs{V})$ 
bits
via every out-port 
$j$.
The leader ignores any message it has received.
Suppose then that a follower $i$ has received a message $\mathsf{m}$.
If this message is the second one that the follower has received, it ignores the message;
otherwise, it sets $\mathsf{id}_i:= \mathsf{m}$
and sends a message $\mathsf{m}\circ j$ via every out-port $j$,
where `$\circ$' means concatenation
(when the follower~$i$ receives multiple messages at once,
the follower arbitrarily breaks the tie and chooses one of them 
as the first message).
Since each message is a sequence of out-port numbers,
this chain of messages uniquely determines a directed path starting from the leader
to each party that receives  one of the messages without ignoring it.
Hence, $\mathsf{id}_i$ is not equal to $\mathsf{id}_j$  whenever $i\neq j$.

The message-passings stop in at most $N$ rounds,
since the number of required rounds
is equal to one plus the length of 
the longest path among those 
determined by the chains of messages,
and any such path
includes each party at most once.
Every party thus moves to the next procedure
after $N$ rounds.
The bit complexity is $O(\abs{V}\abs{E}\log \abs{V})$,
since 
the size of each message is $O(\abs{V}\log \abs{V})$
and each communication link
is used for exactly one message.

To collect all pairs $(\mathsf{id}_i,x_i)$
for $i\in [n]$, all parties run
(a slight modification of)
$\COLORCOUNT$
in Proposition~\ref{pr:ColorCounting}
for $C\deq \set{(\mathsf{id}_i,x_i)\colon i\in [n]}$.
Each party then obtains $C$ 
in $O(N)$ rounds.
Notice that,
unlike the statement of Proposition~\ref{pr:ColorCounting},
the size of $C$ is not constant in this case.
Hence, the bit complexity should be multiplied by at most 
the size of a message:
$O(|V|^2\log |V|)$ 
(since the set $C$ has  $\abs{V}$ pairs of $O(\abs{V}\log \abs{V})$ bits). 
Thus, the bit complexity is $O(N|V|^2\abs{E}\log |V|)$.
\end{proofof}

\begin{proofof}{Theorem~\ref{th:symmetric}}
For each recursion level $t\in [\ceil{\log_2 k}]$,
let $\Xi^{(t)}$ be  the collection of \emph{all possible} equivalence classes
of active parties 
such that each class
$\xi^{(t)}$ in $\Xi^{(t)}$ is the subset of active parties that have
obtained the same sequence $r^{(1)}, \dots, r^{(t)}$
of outcomes in the first through
$t$th levels of recursion
[recall that each outcome is obtained by measruing $\reg{R}_i$ at (modified) STEP~3 or STEP~6].
Note that some $\xi^{(t)}\in \Xi^{(t)}$ may be the empty set.
Define $\Xi^{(0)}\deq\set{\set{i\colon x_i=1}}$.
Since each outcome is a two-bit value,
we have $\abs{ \Xi^{(t)}}=4^t$.
More concretely, $\Xi^{(1)}$ is the finer collection obtained by partitioning $\xi^{(0)}\in \Xi^{(0)}$ into four possible equivalence classes associated with four possible outcomes of $r^{(1)}$: $\set{00, 01, 10, 11}$.
For each $t\in [\ceil{\log_2 k}]$,
let $\varphi^{(t)}$ be a bijection that maps each element in  $\Xi^{(t)}$ to 
the corresponding sequence of outcomes $(r^{(1)}, \dots, r^{(t)})
\in \set{00,01,10,11}^t$. 
We also define $\varphi^{(0)}\colon \xi^{(0)}\mapsto \mathsf{null}$ for the unique element $\xi^{(0)}$ in $\Xi^{(0)}$.
For simplicity, we identify each element $\xi^{(t)}\in \Xi^{(t)}$ with $\varphi^{(t)}(\xi^{(t)})$.

Next, we make a slight modification to Algorithm~QSV as follows
(let {QSV$'$} be the modified version):
If  $\abs{\vect{x}}\ge 2$ or $\abs{\vect{x}}=0$,
then {QSV$'$} outputs $r_i$ at each party $i$,
where $r_i$ is the outcome of measurement on $\reg{R}_i$ made at (modified) STEP~3 or STEP~6 in $Q_{h^*,m^*}$;
if  $\abs{\vect{x}}=1$,  {QSV$'$} outputs ``$\true$'' as the original QSV does.

Now we are ready to present
Algorithm QSYM
for exactly computing a given $f\in \calS_n(k)$.

\begin{algorithm}[t]
\small
\SetAlgoLined
\SetKw{Notation}{Notation}
\KwIn{a classical variable $x_{i}$, 
$k\in \Integer^+$, 
the description of $f\in \calS_k$,
and $N\in \Natural$.}
\KwOut{$\true$ or $\false$.}
\BlankLine
\Begin{
set $x_i(\xi^{(0)}):=x_i$\;
\ForEach{$t:=1,\dots, \ceil{\log _2 k}$}{

\lForEach{$\xi^{(t-1)}\in \Xi^{(t-1)}$}{
	$y_i(\xi^{(t-1)}):= \mbox{QSV$'$}(x_i(\xi^{(t-1)}),N)$
	}\;
	\eIf{$y_i(\xi^{(t-1)})=\true$ for some $\xi^{(t-1)}\in \Xi^{(t-1)}$}
		{
		set $\xi^{(t-1)}_{\min} := \min \set{\xi^{(t-1)} \colon y_i ( \xi^{(t-1)} )  =\true}$\;
		\lIf{$x_i(\xi^{(t-1)}_{\min})=1$}
		{set $S_i:=$``$\mathsf{leader}$''\;}
		\lElse{set $S_i:=$``$\mathsf{follower}$''\;}
		compute $\abs{\vect{x}}$
		by Claim~\ref{cl:eval|x|} with $(x_i,S_i)$\; 
		\Return{$f(\vect{x'})$} for arbitrary $\vect{x'}$ 
		with $\abs{\vect{x'}}=\abs{\vect{x}}$.
		}
		{
		\lForEach{$z\in \set{0,1}^2$ and $\xi^{(t-1)}\in \Xi^{(t-1)}$}
		{set $x_i(\xi^{(t-1)} z):= x_i(\xi^{(t-1)})\wedge [z=y_i(\xi^{(t-1)})]$}.
		}
}
\Return{$f(\vect{x'})$} for arbitrary $\vect{x'}$ with $\abs{\vect{x'}}=k+1$.
}
\SetAlgoRefName{$\mbox{\textbf{QSYM}}$}
\caption{Every party $i$ performs the following operations.} \label{sss}
\end{algorithm}
QSYM 
consists of $\ceil{\log _2 k}$ stages defined as follows:
At stage~$1$,
every party $i$ performs {QSV$'$} with input $(x_i(\xi^{(0)}), N)$,
where $x_i(\xi^{(0)})$ means the input bit $x_i$.
Let $y_i(\xi^{(0)})$ be the output of {QSV$'$}.
If $\abs{\vect{x}}$ is one, then
$y_i(\xi^{(0)})$ is ``$\true$'' by the definiton of  {QSV$'$}.
In this case, only the party $i$ with $x_i=1$ sets $S_i\deq $``$\mathsf{leader}$'',
and then every party can compute $\abs{\vect{x}}$ by
the algorithm in Claim~\ref{cl:eval|x|}, from which
it can compute $f(\vect{x})$ locally.
If $\abs{\vect{x}}\ge 2$ or $\abs{\vect{x}}=0$, then {QSV$'$} 
returns the measurement outcome $y_i(\xi^{(0)})$ to every party $i$.
Every party $i$ then decides which class in $\Xi^{(1)}$
it belongs to by using the value of $y_i(\xi^{(0)})$.
Moreover, for each class $\xi^{(1)}\in \Xi^{(1)}$, 
the party sets the input $x_i(\xi^{(1)})$ for the second stage to $1$
if it is a member of $\xi^{(1)}$
(i.e., $\varphi(\xi^{(1)})=y_i(\xi^{(1)})$),
and to $0$ otherwise.
The algorithm then proceeds to the second stage
to further partition the equivalence classes
(actually, all parties can check whether $\abs{\vect{x}}$ is zero or not
by computing $T_0$ at the beginning of QSYM,
but we design the algorithm as above
just to simplify the descriptions).

More generally,  at each stage $t$,
every party $i$
performs {QSV$'$} with $(x_i(\xi^{(t-1)}), N)$ for each $\xi^{(t-1)}\in \Xi^{(t-1)}$.
If {QSV$'$} returns $y_i(\xi^{(t-1)}) =\true$, then every party 
computes $\abs{\vect{x}}$ by
Claim~\ref{cl:eval|x|}
and outputs $f(\vect{x})$.
Otherwise, {QSV$'$} returns 
$y_i(\xi^{(t-1)})\in \set{0,1}^2$.
For each $z\in  \set{0,1}^2$  and each $\xi^{(t-1)}\in \Xi_{t-1}$,
a unique $\xi^{(t)}\in \Xi_{t}$ satisfies $\xi^{(t)}=\xi^{(t-1)} z$.
Every party $i$ then sets
the input for stage $t+1$ as follows:
\[
x_i(\xi^{(t)})=x_i(\xi^{(t-1)} z) := x_i(\xi^{(t-1)})\wedge [z=y_i(\xi^{(t-1)})],
\]
where  $[z=y_i(\xi^{(t-1)})]$ is the predicate,  which is 1 if and only if 
$z=y_i(\xi^{(t-1)})$;
in other worlds,
for each $\xi^{(t-1)} z\in \Xi_{t}$,
$x_i(\xi^{(t-1)} z)$ is 1 if 
the party $i$ is a member of $\xi^{(t-1)} z$
and $0$ otherwise.

If the algorithm runs up to the  $\ceil{\log _2 k}$-th stage
and  does not output ``$\true$'' for any $\xi^{(\ceil{\log _2 k})}\in \Xi^{(\ceil{\log _2 k})}$,
then $\abs{\vect{x}}$ should be larger than $k$.
Every party thus
chooses an arbitrary $\vect{x'}\in \set{0,1}^N$ with
$\abs{\vect{x'}}>k$,
and computes the value of $f$ on $\vect{x'}$.

For each $t$, all the communication is devoted to
running {QSV$'$} and the algorithm in Claim~\ref{cl:eval|x|},
both of which require $O(N)$ rounds and 
a polynomially bounded number of (qu)bits for communication.
\end{proofof}

\end{document}